\documentclass[12pt]{article}
\usepackage{amsmath}
    \allowdisplaybreaks[4]
\usepackage{graphicx}
\usepackage{enumerate}
\usepackage{natbib}
\usepackage{tikz}
\usetikzlibrary{matrix,chains,scopes,positioning}
\usepackage{url} 

\newcommand{\blind}{0}

\usepackage[utf8]{inputenc} 
\usepackage[T1]{fontenc}    
\usepackage{url}            
\usepackage{booktabs}       
\usepackage{amsfonts}       
\usepackage{nicefrac}       
\usepackage{microtype}      
\usepackage{xcolor}         

\usepackage{amsmath, bm, amsthm, bbm}
\usepackage{array, multirow, graphicx}
\usepackage[normalem]{ulem} 
\usepackage{algorithm}
\usepackage{algpseudocode}
\usepackage{caption}
\usepackage{subcaption}
\usepackage{todonotes}
\usepackage[
    colorlinks=true, 
    linkcolor=blue!70!black, 
    citecolor=blue!70!black,
    urlcolor=blue!70!black,
    breaklinks=true]{hyperref} 
\usepackage{enumitem}

\newtheorem{proposition}{Proposition}
\newtheorem{lemma}{Lemma}
\newtheorem{remark}{Remark}

\DeclareMathOperator{\Var}{Var}

\DeclareMathOperator*{\argmax}{\arg\kern-.1em\max}
\DeclareMathOperator*{\argmin}{\arg\kern-.1em\min}
\begin{document}

\def\spacingset#1{\renewcommand{\baselinestretch}%
{#1}\small\normalsize} \spacingset{1}


\if0\blind
{
  \title{\bf Transfer learning for functional linear regression via control variates
}
  \author{
  	Yuping Yang$^1$\thanks{These authors contributed equally to this work.},
  	Zhiyang Zhou$^2$\footnotemark[1]
    \\
	$^1${\footnotesize School of Mathematics and Physics, Yuxi Normal University}
    \\
	$^2${\footnotesize Joseph J. Zilber College of Public Health, University of Wisconsin-Milwaukee}
	}
  \maketitle
} \fi

\if1\blind
{
  \bigskip
  \bigskip
  \bigskip
  \begin{center}
    {\LARGE\bf Transfer learning for functional linear regression via control variates}
\end{center}
  \medskip
} \fi

\bigskip
\begin{abstract}
Transfer learning (TL) has emerged as a powerful tool for improving 
estimation and prediction performance 
by leveraging information from related datasets,
with the offset TL (O-TL) being a prevailing implementation.
In this paper, 
we adapt the control-variates (CVS) method for TL and
develop CVS-based estimators for scalar-on-function regression,
one of the most fundamental models in functional data analysis.
These estimators rely exclusively on dataset-specific summary statistics, 
thereby avoiding the pooling of subject-level data and 
remaining applicable in privacy-restricted or decentralized settings.
We establish, 
for the first time, 
a theoretical connection between O-TL and CVS-based TL, 
showing that these two seemingly distinct TL strategies adjust local estimators in fundamentally similar ways.
We further derive convergence rates that explicitly account for 
the unavoidable but typically overlooked smoothing error 
arising from discretely observed functional predictors,
and clarify how similarity among covariance functions across datasets governs the performance of TL.
Numerical studies support the theoretical findings and 
demonstrate that 
the proposed methods achieve competitive estimation and prediction performance 
compared with existing alternatives. 
\end{abstract}

\noindent%
\vfill

\newpage
\spacingset{1.9} 

\section{Introduction}\label{sec:intro}

The rapid development and deployment of advanced statistical models across a wide range of real-world applications has highlighted their dependence on sufficient, high-quality training samples. 
This dependence is particularly acute in functional data analysis (FDA), 
where the objects of interest, such as trajectories and images, 
are often treated as infinite-dimensional, 
thereby amplifying the challenges posed by data scarcity.
Consider, for instance, 
the task of modeling health monitoring data 
(e.g., continuous heart rate or oxygen saturation measurements)
for patients with a rare disease,
or stock price trajectories for a small group of corporations.
In such cases, 
the limited availability of relevant data can severely hinder model performance.

A promising solution arises from the well-established concept of transfer learning (TL)
in computer science \citep{PanYang2009, TorreyShavlik2009, WeissKhoshgoftaarWang2016}.
To enhance learning outcomes on a target dataset, 
TL seeks to leverage information from one or more source datasets
that are related to, but not necessarily identical to, the target. 
By repurposing information learned from these sources, 
TL can effectively mitigate the adverse effects of limited target data.

In recent years, 
the statistics community has increasingly embraced TL to improve inference across a variety of models.
As an early contributor to this line of work,
\cite{Bastani2021} proposes a two-stage TL approach for high-dimensional (generalized) linear models using a single source dataset.
In this approach, 
an initial estimator is obtained using pooled source data and 
is subsequently refined using an offset derived from the target data. 
This approach is sometimes referred to as the offset TL (O-TL).
Importantly, 
TL does not guarantee improved performance, 
particularly in multi-source settings.
When source datasets differ substantially from the target, 
TL may even be detrimental,
a phenomenon commonly known as
\textit{negative transfer}
\citep{PanYang2009, TorreyShavlik2009, WeissKhoshgoftaarWang2016}. 
To address this issue,
\cite{LiCaiLi2022a} and \cite{TianFeng2023} extend O-TL 
by incorporating source-detection mechanisms that 
exclude sources likely to induce negative transfer.
O-TL has now been studied under a broad range of models, 
including the Gaussian graphical model,
Cox proportional hazards model,
scalar-on-function regression (SoFR),
and accelerated failure time model;
see, for example,
\cite{HeLiHuLiu2022}, \cite{LiCaiLi2022a}, \cite{LiShenNing2023}, \cite{LinReimherr2024}, 
and \cite{PeiYuShen2025}.

In addition to O-TL,
another strategy for information transfer is the control-variates (CVS) method,
a variance reduction technique originally developed for Monte Carlo simulations.
A key advantage of CVS over O-TL is that 
it avoids pooling individual-level data across sources. 
Instead, 
external sources can contribute via summary statistics, 
enabling information sharing without direct access to subject-level data. 
This distinction is particularly important 
when data sharing is constrained by 
privacy regulations, logistical limitations, or institutional policies. 
Despite this appealing feature,
relatively few studies have employed the CVS method for TL, 
partly because its implementation is less straightforward than that of O-TL.
One notable exception is \cite{DingLiXieWang2024},
which applies CVS to improve the Cox proportional hazards model 
in a multi-source setting while addressing privacy concerns.
As a result, 
substantial methodological and theoretical gaps remain 
in the development of CVS-based TL.

Focusing on TL via CVS,
this work introduces several key innovations.
Firstly, 
we propose two estimators motivated by different perspectives on the CVS method.
One of these estimators incorporates a group lasso penalty to more effectively alleviate negative transfer.
Second,
we establish the theoretical connection between O-TL and CVS in the context of SoFR. 
To the best of our knowledge,
this is the first study to formally bridge these two seemingly distinct TL strategies.
In particular, we point out that,
despite their different underlying principles and algorithmic implementations, 
both strategies adjust \textit{local estimators} 
(defined as estimators constructed using a single dataset)
in fundamentally similar ways.
Third,
we rigorously derive convergence rates that explicitly account for 
smoothing error arising from discretely observed trajectories.
Such error is practically unavoidable in FDA but is often overlooked in existing theoretical studies.
Moreover, 
the derived rates clarify how similarity among covariance functions across datasets governs the performance of TL.
Together, these results provide new insight into how smoothing error and covariance similarity affect both estimation and prediction performance in TL.

The remainder of this paper is organized as follows. 
Section \ref{sec:problem} formally introduces TL for SoFR, 
after reviewing the classical estimation framework.
Section \ref{sec:method} constructs four different estimators,
two based on O-TL and two based on CVS,
highlighting their theoretical connections in Remark \ref{rmk:bridge}.
Section \ref{sec:theory} establishes convergence rates 
explicitly accounting for smoothing error 
and the difference among covariance functions across datasets. 
Section \ref{sec:numerical} presents simulation studies that support the theoretical findings in Section \ref{sec:theory} and compares the proposed estimators using a real-world application.
Finally, Section \ref{sec:conclusion} concludes with a discussion and directions for future research. 
Technical details are deferred to the appendices.

\section{Problem formulation}\label{sec:problem}
We now formally set the stage for TL in the context of SoFR,
one of the most fundamental models in FDA.
Denote the $k$th dataset by 
$\mathcal D^{(k)}=\{(Y_i^{(k)},X_i^{(k)})\}_{i=1}^{n_k}$,
$k=0,1,\ldots,K$.
$Y_i^{(k)}$ are scalar responses,
sharing the identical expectation $\mu_Y^{(k)}$ across $i$.
$X_i^{(k)}$ are realizations of $X^{(k)}$,
a second-order stochastic process on $[0,1]$ with 
mean function $\mu_X^{(k)}=\mu_X^{(k)}(t)={\rm E}X_i^{(k)}(t)$ and 
covariance function 
$C^{(k)}=C^{(k)}(s,t)={\rm cov}(X_i^{(k)}(s),X_i^{(k)}(t))$.
The independence is preserved across $i$ and $k$.
Specifically,
$\mathcal D^{(0)}$ is the target dataset,
while the remaining ones are all source datasets.

Suppose the true model for the $k$th dataset $\mathcal D^{(k)}$ is
\begin{equation}\label{eq:model1}
    Y_i^{(k)}-\mu_Y^{(k)}=\langle X_i^{(k)}-\mu_X^{(k)},\beta^{(k)}\rangle_{L^2}+\varepsilon_i^{(k)}
\end{equation}
in which 
$\langle\cdot, \cdot\rangle_{L^2}$ is the inner product in $L^2([0,1])$.
When $k=0$,
model \eqref{eq:model1} represents the target model;
otherwise, it is refereed to as a source model.  
Coefficient functions $\beta^{(k)}$ are assumed to be square-integrable in $[0,1]$.
Error terms $\varepsilon_i^{(k)}$ are independent with 
mean zero and variance $\sigma_{\varepsilon^{(k)}}^2$.
Of interest is the fixed $\beta^{(0)}\in L^2([0,1])$.
Below we consider model \eqref{eq:model1} with $\mu_X^{(k)}=0$ and $\mu_Y^{(k)}=0$,
because one may always center $Y_i^{(k)}$ and $X_i^{(k)}$ first.
This simplification is minor for our results, but entails a considerable
saving in term of notation.
In practice,
underlying trajectories $X_i^{(k)}$ may not be recorded accurately or continuously.
Instead, 
we collect 
contaminated observations $Z_{i,j}^{(k)}$ discretely at $J_k$ evenly spaced spots 
$t_j^{(k)}=(j-1)/(J_k-1)$, $j=1,\ldots,J_k$, i.e.,
$$
    Z_{i,j}^{(k)}=X_i^{(k)}(t_j^{(k)})+\epsilon_{i,j}^{(k)}
$$
in which $\epsilon_{i,j}^{(k)}$ are measurement errors with 
${\rm E}\epsilon_{i,j}^{(k)}=0$ and 
${\rm var}(\epsilon_{i,j}^{(k)})=\sigma_{\epsilon^{(k)}}^2$.

With $\mathcal D^{(k)}$ only,
one may follow the practical solution \cite[see, e.g.,][]{RamsaySilverman2005}
to generate an estimator of $\beta^{(k)}$ in model \eqref{eq:model1},
say $\hat\beta^{(k)}$.
Basically,
the model fitting is a two-step ridge regression.
First,
smooth the discretely observed $Z_{i,j}^{(k)}$, $j=1,\ldots,J_k$;
namely, approximate the underlying trajectory $X_i^{(k)}$ 
by linear combinations of pre-selected basis functions $\phi_1,\ldots,\phi_M$.
Note that for convenience these basis functions are identical across $k$.
In particular,
$X_i^{(k)}$ is approximated by the smoothed curve
$\widehat X_i^{(k)}=\bm\phi^\top\hat{\bm b}_i^{(k)}$ such that
$$
    (\hat{\bm b}_1^{(k)},\ldots,\hat{\bm b}_{n_k}^{(k)})=
    \argmax_{\bm b_1,\ldots,\bm b_{n_k}\in\mathbb R^M}
    \sum_{i=1}^{n_k}
    (\bm Z_i^{(k)}-\bm\Phi^{(k)}\bm b_i)^\top(\bm Z_i^{(k)}-\bm\Phi^{(k)}\bm b_i)+
    \rho^{(k)}\bm b_i^\top\bm W\bm b_i
$$
for some smoothing parameter $\rho^{(k)}>0$, 
where
\begin{align}
    \notag
    \bm\phi
    &=
    \bm\phi(t)=[\phi_1(t),\ldots,\phi_M(t)]^\top,
    \\\notag
    \bm Z_i^{(k)}
    &=
    [Z_{i,1}^{(k)},\ldots,Z_{i,J_k}^{(k)}]^\top\in\mathbb R^{J_k},
    \\\notag
    \bm\Phi^{(k)}
    &=
    [\phi_m(t_j^{(k)})]_{j=1,\ldots,J_k;m=1,\ldots M}\in\mathbb R^{J_k\times M},
    \quad\text{and}
    \\\label{eq:matrix_W}
    \bm W
    &=
    [\langle\phi_{m_1}'',\phi_{m_2}''\rangle_{L^2}]_{m_1,m_2=1,\ldots,M}\in\mathbb R^{M\times M}.
\end{align}
In \eqref{eq:matrix_W}, 
$\phi_m''$ represents the second-order derivative of $\phi_m$,
$m=1,\ldots,M$. 
Writing 
$$
    \bm P^{(k)}=
    (\bm\Phi^{(k)\top}\bm\Phi^{(k)}+\rho^{(k)}\bm W)^{-1}
    \bm\Phi^{(k)\top}\in\mathbb R^{M\times J_k},
$$
one may immediately point out that
$$
    \hat{\bm b}_i^{(k)}=
    \bm P^{(k)}\bm Z_i^{(k)}
    \in\mathbb R^M.
$$
Consequently,
for the $i$th subject,
the smoothed trajectory is 
$$
    \widehat X_i^{(k)}(t)=\bm\phi^\top(t)\bm P^{(k)}\bm Z_i^{(k)}.
$$

Next, 
introducing
\begin{align}
	\notag
    \bm Z^{(k)}&=[\bm Z_1^{(k)},\ldots,\bm Z_{n_k}^{(k)}]\in\mathbb R^{J_k\times n_k},
    \\\notag
    \bm Y^{(k)}&=[Y_1^{(k)},\ldots,Y_{n_k}^{(k)}]^\top\in\mathbb R^{n_k},
    \\\label{eq:matrix_Psi}
    \bm\Psi &= [\langle\phi_{m_1},\phi_{m_2}\rangle_{L^2}]_{m_1,m_2=1,\ldots,M}\in\mathbb R^{M\times M},
    \quad\text{and}
    \\\notag
    \bm\Omega^{(k)}
    &=\bm\Psi\bm P^{(k)}\bm Z^{(k)}\bm Z^{(k)\top}\bm P^{(k)\top}\bm\Psi\in\mathbb R^{M\times M},
\end{align}
the coefficient function $\beta^{(k)}$ is also estimated by
a linear combination of basis functions:
\begin{equation}\label{eq:hat_beta_k}
    \hat\beta^{(k)}=\bm\phi^\top\hat{\bm c}^{(k)}
\end{equation}
in which
\begin{align*}
    \hat{\bm c}^{(k)}
    &=
    \argmax_{\bm c\in\mathbb R^M}
    \sum_{i=1}^{n_k}|Y_i^{(k)}-\langle\widehat X_i^{(k)},\bm\phi^\top\bm c\rangle_{L^2}|^2+\lambda^{(k)}\bm c^\top \bm W\bm c
    \\
    &=
    (\bm\Omega^{(k)}+\lambda^{(k)}\bm W)^{-1}
    \bm\Psi\bm P^{(k)}\bm Z^{(k)}\bm Y^{(k)}
\end{align*}
for another smoothing parameter $\lambda^{(k)}>0$.
Employing the $k$th dataset $\mathcal D^{(k)}$ only,
$\hat\beta^{(k)}$ in \eqref{eq:hat_beta_k} is considered as a local estimator,
$k=0,\ldots,K$.
Thus,
TL is the procedure to improve $\hat\beta^{(0)}$,
borrowing strength from the $K$ source datasets.

\section{Methodology}\label{sec:method}

\subsection{Offset transfer learning}\label{sec:otl}

Although our focus is on TL via CVS, 
we begin by outlining its alternative, O-TL, 
which are available when individual-level data are shareable across sources.
Suppose there is a transferable set
\begin{equation}\label{eq:transferable_set}
    \mathcal K=\{k\in\{1,\ldots,K\}: \mathcal D^{(k)}\text{ that is unlikely to yield the negative transfer}\}.
\end{equation}
Roughly speaking,
$\mathcal K$ is a subset of indices of source models/datasets that are similar to the target one.
If $\mathcal K$ is known to be $\{k_1,\ldots,k_{|\mathcal K|}\}$
with $|\mathcal K|$ denoting the cardinality of set $\mathcal K$,
then one may pool together all the smoothed curves and corresponding responses
from datasets belonging to $\mathcal K$.
A centered source estimator follows,
serving as the initial guess on $\beta^{(0)}$:
\begin{equation}\label{eq:hat_beta_K}
    \hat\beta^{(\mathcal K)}=\bm\phi^\top\hat{\bm c}^{(\mathcal K)}
\end{equation}
with
\begin{align*}
    \hat{\bm c}^{(\mathcal K)}
    &=
    \argmin_{\bm c\in\mathbb R^M}
    \sum_{k\in\mathcal K}
    \sum_{i=1}^{n_k}
    |Y_i^{(k)}-\langle\widehat X_i^{(k)},\bm\phi^\top\bm c\rangle_{L^2}|^2+
    \lambda^{(\mathcal K)}\bm c^\top \bm W\bm c
    \\
    &=\left(
    	\bm V^{(\mathcal K)\top}\bm V^{(\mathcal K)}+
    	\lambda^{(\mathcal K)}\bm W\right)^{-1}
    \bm V^{(\mathcal K)\top}\bm Y^{(\mathcal K)},
\end{align*}
where $\lambda^{(\mathcal K)}>0$ is a smoothing parameter,
$\bm V^{(\mathcal K)}=[\bm\Psi\bm P^{(k_1)}\bm Z^{(k_1)},\ldots,\bm\Psi\bm P^{(k_{|\mathcal K|})}\bm Z^{(k_{|\mathcal K|})}]^\top\in\mathbb R^{(\sum_{k\in\mathcal K}n_k)\times M}$
and $\bm Y^{(\mathcal K)}=[\bm Y^{(k_1)\top},\ldots,\bm Y^{(k_{|\mathcal K|})\top}]^\top\in\mathbb R^{\sum_{k\in\mathcal K}n_k}$.
O-TL is finalized by 
imposing an offset $\bm o^{(\mathcal K)}\in\mathbb R^M$ to given $\hat{\bm c}^{(\mathcal K)}$;
namely, the resulting estimator is
\begin{equation}\label{eq:hat_beta_OTL}
    \hat\beta_{O,\mathcal K}^{(0)}=\bm\phi^\top(\hat{\bm c}^{(\mathcal K)}+\bm o^{(\mathcal K)}),
\end{equation}
where, for $\lambda_{\bm O}^{(\mathcal K)}>0$,
\begin{align*}
    \bm o^{(\mathcal K)}
    &=
    \argmin_{\bm o\in\mathbb R^M}
    \sum_{i=1}^{n_0}
    |Y_i^{(0)}-\langle\widehat X_i^{(0)},\bm\phi^\top(\hat{\bm c}^{(\mathcal K)}+\bm o)\rangle_{L^2}|^2+
    \lambda_{\bm O}^{(\mathcal K)}\bm o^\top \bm W\bm o
    \\
    &=(\bm\Omega^{(0)}+
    \lambda_{\bm O}^{(\mathcal K)}\bm W)^{-1}
    \bm\Psi\bm P^{(0)}\bm Z^{(0)}
    (\bm Y^{(0)}-\bm Z^{(0)\top}\bm P^{(0)\top}\bm\Psi\hat{\bm c}^{(\mathcal K)}).
\end{align*}

\begin{algorithm}[!h]
\caption{O-TL for SoFR with $\mathcal K$ known to be $\{k_1,\ldots,k_{|\mathcal K|}\}$}\label{algo:OTL}
\begin{algorithmic}[1]
   \State Input
   basis $\bm\phi$ and
   contaminated paired observations
   $\{(\bm Z_i^{(k)}, Y_i^{(k)})\}_{i=1}^{n_k}$,
   $k\in\{0\}\bigcup\mathcal K$.
   \State Fixing $k\in\mathcal K$ 
   and following Section \ref{sec:problem}, 
   smooth $\bm Z_i^{(k)}$ and generate smoothed curves $\widehat X_i^{(k)}$, $i=1,\ldots,n_k$.
   \State Generate the centered source estimator 
   $\hat\beta^{(\mathcal K)}$ in \eqref{eq:hat_beta_K}.
   \State 
   Compute the offset $\bm o^{\mathcal K}$
   and output 
   $\hat\beta_{O,\mathcal K}^{(0)}$
   in \eqref{eq:hat_beta_OTL}.
\end{algorithmic}
\end{algorithm}

We understand that
Algorithm \ref{algo:OTL} is typically impractical,
because there is little prior knowledge on $\mathcal K$ in \eqref{eq:transferable_set}
in real-world applications.
Simply taking $\mathcal K=\{1,\ldots,K\}$ could be detrimental:
Algorithm \ref{algo:OTL} tends to treat all the sources in a fair way
(see Remark \ref{rmk:bridge} on Page \pageref{rmk:bridge}).
As a result,
a source model that is dissimilar to the target model could still be unfairly assigned a high weight by Algorithm \ref{algo:OTL}.
To handle the unknown $\mathcal K$,
\cite{LiCaiLi2022a} and \cite{LinReimherr2024} suggest
the Aggregation-based O-TL (AO-TL).
The idea is to first construct
a set of candidates for $\mathcal K$, 
say $\mathbb K$,
such that there exists at least one entry of $\mathbb K$ equal to true $\mathcal K$ 
with high probability.
For each entry of $\mathbb K$,
one may yield a candidate estimator via Algorithm \ref{algo:OTL}.
Let $\mathbb B$ denote the set of all these candidate estimators.
AO-TL is finalized by aggregating elements of $\mathbb B$.
There are several approaches available for this purpose;
following \cite{LinReimherr2024},
Algorithm \ref{algo:AOTL} adopts the hyper-sparse aggregation \citep{GaiffasLecue2011},
which aggregates only two entries of $\mathbb B$ based upon
\begin{align}
	\label{eq:pred_err}
    R_A^{(1)}(\beta)
    &=
    \frac{1}{|A|}\sum_{i\in A}(Y_i^{(0)}-\langle\widehat X_i^{(0)},\beta\rangle_{L^2})^2,
    \\ \label{eq:est_err}
    R_A^{(2)}(\beta_1,\beta_2)
    &=
    \frac{1}{|A|}\sum_{i\in A}
    \langle\widehat X_i^{(0)},\beta_1-\beta_2\rangle_{L^2}^2,
\end{align}
with $A$ as a set of indices and $\beta$'s as candidate estimators.
The hyper-sparse aggregation helps alleviate the impact of 
source models that significantly differ from the target model.

\begin{algorithm}[!htp]
\caption{AO-TL for SoFR}\label{algo:AOTL}
\begin{algorithmic}[1]
   \State Input
   basis $\bm\phi$,
   contaminated paired observations
   $\{(\bm Z_i^{(k)}, Y_i^{(k)})\}_{i=1}^{n_k}$,
   $k=0,\ldots,K$, and
   pre-specified parameters $b_1$ and $b_2$;
   see Appendix \ref{sec:constants} for details on $b_1$ and $b_2$.
   \State Fixing $k\in\{1,\ldots,K\}$ 
   and following Section \ref{sec:problem}, 
   smooth $\bm Z_i^{(k)}$ and generate smoothed curves $\widehat X_i^{(k)}$, $i=1,\ldots,n_k$.
   Then give local estimators $\hat\beta^{(k)}$ in \eqref{eq:hat_beta_k}, $k=1,\ldots,K$.
   \State Randomly split the indices of $\mathcal D^{(0)}$ into two equal-sized subsets, 
   say $\mathcal D^{(0_1)}$ and $\mathcal D^{(0_2)}$.
   Generate $\hat\beta^{(0_1)}$ 
   which is constructed as \eqref{eq:hat_beta_k}
   but merely utilizes data points in $\mathcal D^{(0_1)}$.
   \State 
   Recall $R_A^{(2)}(\cdot,\cdot)$ defined as \eqref{eq:est_err}.
   Construct $\mathcal K_k$, the $k$th candidate for $\mathcal K$, 
   such that 
   \begin{equation}\label{eq:K_k}
       \mathcal K_k=\left\{
       m\in\{1,\ldots, K\}:
       R_{\mathcal D^{(0_1)}}^{(2)}(\hat\beta^{(0_1)},\hat\beta^{(m)})
       \text{ is among the first } k\text{ smallest of all}
    \right\}.
   \end{equation}
   \State For each $k\in\{1,\ldots,K\}$,
   implement Algorithm \ref{algo:OTL} by taking $\mathcal K_k$ as the transferable set
   and $\mathcal D^{(0_1)}$ as the target dataset.
   Denote by $\hat\beta_{O,\mathcal K_k}^{(0_1)}$
   the resulting estimator.
   \State 
   Further split $\mathcal D^{(0_2)}$ randomly into 
   equal-sized $\mathcal D^{(0_{21})}$ and $\mathcal D^{(0_{22})}$.
   Recall $R_A^{(1)}(\cdot)$ in \eqref{eq:pred_err}.
   Let $\hat\beta_{O,\mathcal K_{k^*}}^{(0_1)}=
   \argmin_{\beta\in\mathbb B} R_{\mathcal D^{(0_{21})}}^{(1)}(\beta)$
   with $\mathbb B=\{\hat\beta^{(0_1)},\hat\beta_{O,\mathcal K_1}^{(0_1)},\ldots,\hat\beta_{O,\mathcal K_K}^{(0_1)}\}$. 
   A subset of $\mathbb B$ follows:
   $$
   \mathbb B_1=\bigg\{\beta\in\mathbb B:
   R_{\mathcal D^{(0_{21})}}^{(1)}(\beta)\leq
   R_{\mathcal D^{(0_{21})}}^{(1)}(\hat\beta_{O,\mathcal K_{k^*}}^{(0_1)})+
   b_1\Big(
   	b_2^2\vee b_2
	\sqrt{
	   	R_{\mathcal D^{(0_{21})}}^{(2)}(\hat\beta_{O,\mathcal K_{k^*}}^{(0_1)},\theta)
   	}\Big)\bigg\},
   $$
   with $a\vee b=\max\{a,b\}$ and $a\wedge b=\min\{a,b\}$.
   See Appendix \ref{sec:constants} for $b_1$ and $b_2$.
   \State Return
   \begin{align}
   	\hat\beta_{AO}^{(0)}
   	=\argmin_{\beta\in\mathbb B_2}
   	R_{\mathcal D^{(0_{22})}}(\beta)
   	\label{eq:hat_beta_AOTL}
   	=\argmin_{\beta\in\mathbb B_3}
       R_{\mathcal D^{(0_{22})}}(\beta)
   \end{align}
   in which 
   $\mathbb B_2=\{a\hat\theta_{O,\mathcal K_{k^*}}^{(0_1)}+(1-a)\theta:\theta\in\mathbb B_1, a\in[0,1]\}$
   and
   $\mathbb B_3=\{a_\theta\hat\beta_{O,\mathcal K_{k^*}}^{(0_1)}+(1-a_\theta)\theta:
		a_\theta=
		0\vee
		\Big[\frac{1}{2}
		\big\{
		   	R_{\mathcal D^{(0_{22})}}(\theta)-
		   	R_{\mathcal D^{(0_{22})}}(\hat\beta_{O,\mathcal K_{k^*}}^{(0_1)})
		\big\}/
		R_{\mathcal D^{(0_{22})}}^{(2)}(\hat\beta_{O,\mathcal K_{k^*}}^{(0_1)},\theta)
		+\frac{1}{2}
		\Big]
		\wedge 1
	\}$.
	See \citet[pp.~1817]{GaiffasLecue2011} for the proof of the second equation of \eqref{eq:hat_beta_AOTL}.
\end{algorithmic}
\end{algorithm}

\subsection{Control-variates method}

In the presence of privacy constraints, 
sharing individual-level data is typically prohibited. 
As a result, 
O-TL becomes unusable, 
since it requires pooling individual-level data across sources.
In contrast,
the information transfer via CVS method remains applicable 
because it relies solely on dataset-specific summary statistics.

Define $\hat{\bm\delta}^{(k)}=\hat{\bm c}^{(0)}-\hat{\bm c}^{(k)}\in\mathbb R^M$ 
for $k=1,\ldots,K$
and consider a linear combination of the $M$ basis functions
$\phi_1,\ldots,\phi_M$,
say $\bm\phi^{\top}\hat{\bm c}_{\bm U,\bm\delta}^{(0)}$,
such that
\begin{equation}\label{eq:hat_c_0}
    \hat{\bm c}_{\bm U,\bm\delta}^{(0)}= 
    \hat{\bm c}^{(0)}-\bm U(\hat{\bm\delta}-\bm\delta),
\end{equation}
where 
\begin{equation}\label{eq:ctrl_var}
    \hat{\bm\delta}=
    [\hat{\bm\delta}^{(1)\top},\ldots,\hat{\bm\delta}^{(K)\top}]^\top=
    \mathbf 1_K\otimes\hat{\bm c}^{(0)}-
    [\hat{\bm c}^{(1)\top},\ldots,\hat{\bm c}^{(K)\top}]^\top
    \in\mathbb R^{MK}
\end{equation}
is referred to as the control variates.
$\mathbf 1_K$ and $\otimes$ denote
the $K$-vector of ones and Kronecker product, respectively.
$\hat{\bm c}_{\bm U,\bm\delta}^{(0)}$ in \eqref{eq:hat_c_0} 
is a function of 
$\bm U\in\mathbb R^{M\times MK}$ and $\bm\delta\in\mathbb R^{MK}$,
where both $\bm U$ and $\bm\delta$ are assumed to be either non-random or determined solely by
$\mathcal Z=\{\bm Z^{(0)},\ldots,\bm Z^{(K)}\}$.
Apparently,
${\rm E}(\bm\phi^{\top}(t)\hat{\bm c}_{\bm U,\bm\delta}^{(0)}\mid\mathcal Z)=
{\rm E}(\hat\beta^{(0)}(t)\mid\mathcal Z)$ 
when $\bm\delta={\rm E}(\hat{\bm\delta}\mid\mathcal Z)$.
However, 
$\bm\phi^{\top}\hat{\bm c}_{\bm U,\bm\delta}^{(0)}$ potentially enjoys less variation in the sense that
$\Var(\bm\phi^{\top}(t)\hat{\bm c}_{\bm U,\bm\delta}^{(0)}\mid\mathcal Z)
\leq\Var(\hat\beta^{(0)}(t)\mid\mathcal Z)$
pointwisely, 
provided that $\bm U$ is chosen appropriately.
Actually,
free of the value of $\bm\delta$,
\begin{multline*}
    {\rm var}(\hat{\bm c}_{\bm U,\bm\delta}^{(0)}\mid\mathcal Z)
    =
    {\rm var}(\hat{\bm c}^{(0)}\mid\bm Z^{(0)})
    \\
    -\big\{\mathbf 1_K^\top\otimes{\rm var}(\hat{\bm c}^{(0)}\mid\bm Z^{(0)})\big\}\bm U^\top-
    \bm U\big\{\mathbf 1_K\otimes{\rm var}(\hat{\bm c}^{(0)}\mid\bm Z^{(0)})\big\}+
    \bm U{\rm var}(\hat{\bm\delta}\mid\mathcal Z)\bm U^\top
\end{multline*}
is minimized when $\bm U$ takes
\begin{align}
    \notag
    \bm U^*
    &=
    \{\mathbf 1_K^\top\otimes{\rm var}(\hat{\bm c}^{(0)}\mid\bm Z^{(0)})\}
    \{{\rm var}(\hat{\bm\delta}\mid\mathcal Z)\}^{-1}
    \\ \label{eq:U_star}
    &=
    \bigg\{
    \sum_{k=0}^K
    {\rm var}^{-1}(\hat{\bm c}^{(k)}\mid\bm Z^{(k)})
    \bigg\}^{-1}
    [
    {\rm var}^{-1}(\hat{\bm c}^{(1)}\mid\bm Z^{(1)}),
    \ldots,
    {\rm var}^{-1}(\hat{\bm c}^{(K)}\mid\bm Z^{(K)})
    ],
\end{align}
i.e., 
${\rm var}(\hat{\bm c}_{\bm U,\bm\delta}^{(0)}\mid\mathcal Z)-
{\rm var}(\hat{\bm c}_{\bm U^*,\bm\delta}^{(0)}\mid\mathcal Z)$
is positive semidefinite for all $\bm U$ and $\bm\delta$;
see Appendix \ref{sec:derive_U_star} for how to derive $\bm U^*$ in \eqref{eq:U_star}.
It is hence reasonable to set $\bm U=\bm U^*$ and $\bm\delta={\rm E}(\hat{\bm\delta}\mid\mathcal Z)$
in \eqref{eq:hat_c_0},
leading to
\begin{align}
    \label{eq:hat_c_0_Ustar}
    \hat{\bm c}_{\bm U^*,{\rm E}(\hat{\bm\delta}\mid\mathcal Z)}^{(0)}
    = 
    \hat{\bm c}^{(0)}-\bm U^*\{\hat{\bm\delta}-{\rm E}(\hat{\bm\delta}\mid\mathcal Z)\}
\end{align}
Introducing
\begin{equation}\label{eq:tilde_beta_0}
    \tilde\beta^{(0)}=
    \bm\phi^\top\hat{\bm c}_{\bm U^*,{\rm E}(\hat{\bm\delta}\mid\mathcal Z)}^{(0)},
\end{equation}
it is clear that
${\rm var}(\hat\beta^{(0)}(t)\mid\mathcal Z)\geq{\rm var}(\tilde\beta^{(0)}(t)\mid\mathcal Z)$
and 
${\rm E}(\hat\beta^{(0)}(t)\mid\mathcal Z)={\rm E}(\tilde\beta^{(0)}(t)\mid\mathcal Z)$
pointwisely.

\begin{remark}
    Compared with
    $\hat{\bm c}^{(0)}={\rm E}(\hat{\bm c}^{(0)}\mid\bm Z^{(0)})+\{\hat{\bm c}^{(0)}-{\rm E}(\hat{\bm c}^{(0)}\mid\bm Z^{(0)})\}$,
    $\hat{\bm c}_{\bm U^*,{\rm E}(\hat{\bm\delta}\mid\mathcal Z)}^{(0)}$
    in \eqref{eq:hat_c_0_Ustar}
    adjusts ${\rm E}(\hat{\bm c}^{(0)}\mid\bm Z^{(0)})$
    with a weighted average of  discrepancies
    $\hat{\bm c}^{(k)}-{\rm E}(\hat{\bm c}^{(k)}\mid\bm Z^{(k)})$,
    $k=0,\ldots,K$,
    instead of merely the gap between
    $\hat{\bm c}^{(0)}$ and ${\rm E}(\hat{\bm c}^{(0)}\mid\bm Z^{(0)})$.
    This is how
    $\tilde\beta^{(0)}$ in \eqref{eq:tilde_beta_0}
    improves $\hat\beta^{(0)}$ in \eqref{eq:hat_beta_k}.
\end{remark}

The information transfer via CVS is finalized by substituting 
$\widehat{\rm E}(\hat{\bm c}^{(k)}\mid\bm Z^{(k)})$
and
$\widehat{\rm var}(\hat{\bm c}^{(k)}\mid\bm Z^{(k)})$,
$k=0,\ldots,K$, 
all specified in Appendix \ref{sec:approx_exp_var},
for ${\rm E}(\hat{\bm c}^{(k)}\mid\bm Z^{(k)})$ and 
${\rm var}(\hat{\bm c}^{(k)}\mid\bm Z^{(k)})$ in \eqref{eq:hat_c_0_Ustar},
respectively.
The resulting estimator of $\beta^{(0)}$ is
\begin{align}
    \label{eq:hat_beta_C}
    \hat\beta_C^{(0)}
    =
    \bm\phi^\top
    \hat{\bm c}_{\widehat{\bm U}^*,\widehat{\rm E}(\hat{\bm\delta}\mid\mathcal Z)}^{(0)}
    = 
    \hat{\bm c}^{(0)}-\widehat{\bm U}^*\{\hat{\bm\delta}-\widehat{\rm E}(\hat{\bm\delta}\mid\mathcal Z)\},
\end{align}
where
\begin{equation*}\label{eq:hat_U_star}
    \widehat{\bm U}^*
    =
    \bigg\{
    \sum_{k=0}^K
    \widehat{\rm var}^{-1}(\hat{\bm c}^{(k)}\mid\bm Z^{(k)})
    \bigg\}^{-1}
    [
    \widehat{\rm var}^{-1}(\hat{\bm c}^{(1)}\mid\bm Z^{(1)}),
    \ldots,
    \widehat{\rm var}^{-1}(\hat{\bm c}^{(K)}\mid\bm Z^{(K)})
    ]
\end{equation*}
and
$$
    \widehat{\rm E}(\hat{\bm\delta}\mid\mathcal Z)=
    \mathbf 1_K\otimes\widehat{\rm E}(\hat{\bm c}^{(0)}\mid\bm Z^{(0)})-
    \big[
    \{\widehat{\rm E}(\hat{\bm c}^{(1)}\mid\bm Z^{(1)})\}^\top,\ldots,
    \{\widehat{\rm E}(\hat{\bm c}^{(K)}\mid\bm Z^{(K)})\}^\top
    \big]^\top.
$$
As a surrogate of $\tilde\beta^{(0)}$ in \eqref{eq:tilde_beta_0},
$\hat\beta_C^{(0)}$ in \eqref{eq:hat_beta_C}
converges to $\tilde\beta^{(0)}$ as sample sizes $n_k$ all diverge;
see Proposition \ref{prop:hat_beta_c_estimation_err} for details.

\begin{algorithm}[ht]
\caption{CVS for SoFR}\label{algo:CVS}
\begin{algorithmic}[1]
   \State Input
   basis $\bm\phi$ and
   contaminated paired observations
   $\{(\bm Z_i^{(k)}, Y_i^{(k)})\}_{i=1}^{n_k}$,
   $k=0,\ldots,K$.
   \State Fixing $k$ 
   and following Section \ref{sec:problem}, 
   smooth $\bm Z_i^{(k)}$ and generate smoothed curves $\widehat X_i^{(k)}$, $i=1,\ldots,n_k$.
   \State For the $k$th dataset,
   generate $\hat{\bm c}^{(k)}$ in \eqref{eq:hat_beta_k},
   $k=0,\ldots,K$.
   \State Approximate 
   ${\rm E}(\hat{\bm c}^{(k)}\mid\bm Z^{(k)})$
   and ${\rm var}(\hat{\bm c}^{(k)}\mid\bm Z^{(k)})$
   with $\widehat{\rm E}(\hat{\bm c}^{(k)}\mid\bm Z^{(k)})$
   and $\widehat{\rm var}(\hat{\bm c}^{(k)}\mid\bm Z^{(k)})$,
   respectively,
   following Appendix \ref{sec:approx_exp_var}, $k=0,\ldots,K$.
   \State Output $\hat\beta_C^{(0)}$ in \eqref{eq:hat_beta_C}.
\end{algorithmic}
\end{algorithm}

\subsection{Another look at the control-covariates method}

We begin with the following quadratic loss function
\begin{equation}\label{eq:l_c_delta}
    \ell(\bm c,\bm\delta)
    =
    \begin{bmatrix}
        \hat{\bm c}^{(0)}-\bm c\\
        \hat{\bm\delta}-\bm\delta
    \end{bmatrix}^\top
    {\rm cov}^{-1}(\hat{\bm c}^{(0)}, \hat{\bm\delta}\mid\mathcal Z)
    \begin{bmatrix}
        \hat{\bm c}^{(0)}-\bm c\\
        \hat{\bm\delta}-\bm\delta
    \end{bmatrix},
\end{equation}
where 
${\rm cov}^{-1}(\hat{\bm c}^{(0)}, \hat{\bm\delta}\mid\mathcal Z)$
is specified in Appendix \ref{sec:inverse_matrix}.
The unconstrained minimizer of \eqref{eq:l_c_delta} is simply $(\hat{\bm c}^{(0)},\hat{\bm\delta})$,
corresponding to local estimator $\hat\beta^{(0)}$ in \eqref{eq:hat_beta_k}.
Meanwhile,
fixing $\bm\delta={\rm E}(\hat{\bm\delta}\mid\mathcal Z)\in\mathbb R^{MK}$,
the minimizer of \eqref{eq:l_c_delta} 
with respect to $\bm c\in\mathbb R^M$ coincides with 
$\hat{\bm c}_{\bm U^*,{\rm E}(\hat{\bm\delta}\mid\mathcal Z)}^{(0)}$
in \eqref{eq:hat_c_0_Ustar},
corresponding to $\tilde\beta^{(0)}$ in \eqref{eq:tilde_beta_0}.
The unknown ${\rm E}(\hat{\bm\delta}\mid\mathcal Z)$
and block-wise nature of $\hat{\bm\delta}$
motivates us to minimize the empirical counterpart of \eqref{eq:l_c_delta}
with the group lasso penalty \citep{YuanLin2006} attached:
\begin{equation}\label{eq:pl_c_delta}
    \widehat{p\ell}(\bm c,\bm\delta)=
    \begin{bmatrix}
        \hat{\bm c}^{(0)}-\bm c\\
        \hat{\bm\delta}-\bm\delta
    \end{bmatrix}^\top
    \widehat{\rm cov}^{-1}(\hat{\bm c}^{(0)}, \hat{\bm\delta}\mid\mathcal Z)
    \begin{bmatrix}
        \hat{\bm c}^{(0)}-\bm c\\
        \hat{\bm\delta}-\bm\delta
    \end{bmatrix}+
    \zeta\sum_{k=1}^K\|\bm\delta^{(k)}\|_2,
\end{equation}
where $\widehat{\rm cov}^{-1}(\hat{\bm c}^{(0)}, \hat{\bm\delta}\mid\mathcal Z)$
is given in Appendix \ref{sec:inverse_matrix} and
$\bm\delta^{(k)}\in\mathbb R^M$ consists of 
the entries from the ($(k-1)M+1$)-th to the ($kM$)-th position of $\bm\delta$,
i.e.,
$\bm\delta=[\bm\delta^{(1)\top},\ldots,\bm\delta^{(K)\top}]^\top$.
$\|\cdot\|_2$ here denotes the spectral norm for matrices;
that is,
for a real matrix $\bm A$,
$\|\bm A\|_2$ is the largest singular value of $\bm A$.
For vectors,
$\|\cdot\|_2$ reduces to the Euclidean norm.

Due to its convexity,
$\widehat{p\ell}(\bm c,\bm\delta)$ in \eqref{eq:pl_c_delta} can be minimized in a block-wise manner.
Specifically,
we decompose the minimization of $\widehat{p\ell}(\bm c,\bm\delta)$ \eqref{eq:pl_c_delta}
into the following three substeps.
\begin{enumerate}[label=S\arabic*.]
    \item Fixing arbitrary value of $\bm\delta\in\mathbb R^{MK}$,
    	the minimizer of \eqref{eq:pl_c_delta} with respect to $\bm c$
    	is exactly
        \begin{align}
            \label{eq:c_delta}
            \hat{\bm c}_{\widehat{\bm U}^*,\bm\delta}^{(0)}
            =
            \hat{\bm c}^{(0)}-\widehat{\bm U}^*(\hat{\bm\delta}-\bm\delta).
        \end{align}
    \item
    	Plugging $\bm c=\hat{\bm c}_{\bm U^*,\bm\delta}^{(0)}$ back into \eqref{eq:pl_c_delta},
        the minimization of \eqref{eq:pl_c_delta} is equivalent to locating
        \begin{equation}\label{eq:delta_star}
            \bm{\delta}^\zeta
            =\argmin_{\bm\delta\in\mathbb R^{MK}}
            (\hat{\bm\delta}-\bm\delta)^\top
            \widehat{\rm var}^{-1}(\hat{\bm\delta}\mid\mathcal Z)
            (\hat{\bm\delta}-\bm\delta)+
            \zeta\sum_{k=1}^K\|\bm\delta^{(k)}\|_2,
        \end{equation}
        where
        $\widehat{\rm var}^{-1}(\hat{\bm\delta}\mid\bm Z^{(0)})=\widehat{\bm B}_1-\widehat{\bm B}_2$
        is the empirical version of ${\rm var}^{-1}(\hat{\bm\delta}\mid\bm Z^{(0)})$
        (as elaborated in Appendix \ref{sec:derive_U_star}) with
        $$
           \widehat{\bm B}_1={\rm diag}\{
           \widehat{\rm var}^{-1}(\hat{\bm c}^{(1)}\mid\bm Z^{(1)}),\ldots,
           \widehat{\rm var}^{-1}(\hat{\bm c}^{(K)}\mid\bm Z^{(K)})
           \}
        $$
        and
        \begin{align*}
            \widehat{\bm B}_2
            &=
            [
            \widehat{\rm var}^{-1}(\hat{\bm c}^{(1)}\mid\bm Z^{(1)}),
            \ldots,
            \widehat{\rm var}^{-1}(\hat{\bm c}^{(K)}\mid\bm Z^{(K)})
            ]^\top
            \\
            &\qquad
            \bigg\{
            \sum_{k=0}^K
            \widehat{\rm var}^{-1}(\hat{\bm c}^{(k)}\mid\bm Z^{(k)})
            \bigg\}^{-1}
            [
            \widehat{\rm var}^{-1}(\hat{\bm c}^{(1)}\mid\bm Z^{(1)}),
            \ldots,
            \widehat{\rm var}^{-1}(\hat{\bm c}^{(K)}\mid\bm Z^{(K)})
            ].
        \end{align*}
        In practice,
        one may select $\zeta$ and then give $\bm{\delta}^\zeta$ via existing software packages, e.g., 
        \texttt{R} packages \texttt{gglasso} \citep{YangZouBhatnagar2024}
        and \texttt{sparsegl} \citep{LiangEtal2024}.
    \item The minimization of $\widehat{p\ell}(\bm c,\bm\delta)$ is finalized by
    	substituting $\bm{\delta}^\zeta$ \eqref{eq:delta_star} for $\bm\delta$ in \eqref{eq:c_delta}.
\end{enumerate}

Given the minimizer of $\widehat{p\ell}(\bm c,\bm\delta)$, 
i.e., $(\hat{\bm c}_{\widehat{\bm U}^*,\bm{\delta}^\zeta}^{(0)},\bm{\delta}^\zeta)$,
another improvement of $\hat\beta^{(0)}$ in \eqref{eq:hat_beta_k} follows:
\begin{equation}\label{eq:hat_beta_PC}
    \hat\beta_{PC}^{(0)}=
    \bm\phi^\top
    \hat{\bm c}_{\widehat{\bm U}^*,{\bm\delta}^\zeta}^{(0)}=
    \bm\phi^\top
    \{\hat{\bm c}^{(0)}-\widehat{\bm U}^*(\hat{\bm\delta}-{\bm\delta}^\zeta)\}.
\end{equation}
The information transfer in \eqref{eq:hat_beta_PC} parallels that of CVS but incorporates an additional penalty term. 
We therefore refer to it as the penalized CVS (pCVS).

\begin{algorithm}[ht]
\caption{pCVS for SoFR}\label{algo:pCVS}
\begin{algorithmic}[1]
   \State Input
   basis $\bm\phi$ and
   contaminated paired observations
   $\{(\bm Z_i^{(k)}, Y_i^{(k)})\}_{i=1}^{n_k}$,
   $k=0,\ldots,K$.
   \State Fixing $k$ 
   and following Section \ref{sec:problem}, 
   smooth $\bm Z_i^{(k)}$ and generate smoothed curves $\widehat X_i^{(k)}$, $i=1,\ldots,n_k$.
   \State Compute the control variates $\hat{\bm\delta}$ \eqref{eq:ctrl_var}.
   \State Approximate 
   ${\rm var}(\hat{\bm c}^{(k)}\mid\bm Z^{(k)})$
   with $\widehat{\rm var}(\hat{\bm c}^{(k)}\mid\bm Z^{(k)})$
   following Appendix \ref{sec:approx_exp_var}, $k=0,\ldots,K$.
   \State Numerically figure out $\bm{\delta}^\zeta$ \eqref{eq:delta_star}.
   \State Output $\hat\beta_{PC}^{(0)}$ \eqref{eq:hat_beta_PC}.
\end{algorithmic}
\end{algorithm}

\begin{remark}\label{rmk:bridge}
    Although CVS and pCVS
    seem different from 
    O-TL (with known transferable set $\mathcal K=\{k_1,\ldots,k_{|\mathcal K|}\}$)
    and
    AO-TL (with unknown transferable set)
    in Section \ref{sec:otl}, 
    we can still bridge these four
    in a non-rigorous manner.
    Notably, in \eqref{eq:hat_beta_OTL},
    $$
        \hat{\bm c}^{(\mathcal K)}+\bm o^{(\mathcal K)}=
        \hat{\bm c}^{(0)}
        -
        (\bm\Omega^{(0)}+\lambda_{\bm O}^{(\mathcal K)}\bm W)^{-1}
        \big[
        \big\{
        \bm\Omega^{(0)}+
        (\lambda_{\bm O}^{(\mathcal K)}-\lambda^{(0)})\bm W
        \big\}
        \hat{\bm c}^{(0)}
        +
        \lambda_{\bm O}^{(\mathcal K)}\bm W\hat{\bm c}^{(\mathcal K)}
        \big].
    $$
    As pointed out by
    \citet[Corollary 3]{ZhangDuchiWainwright2015},
    $\hat{\bm c}^{(\mathcal K)}\approx|\mathcal K|^{-1}\sum_{k\in\mathcal K}\hat{\bm c}^{(k)}$,
    if $\lambda^{(k_1)}=\cdots=\lambda^{(k_{|\mathcal K|})}$ and $n_{k_1}=\cdots=n_{k_{|\mathcal K|}}$ with relatively small $|\mathcal K|$ and 
    considerably small $\max_{k\in\mathcal K} M\sigma_{\varepsilon^{(k)}}^2/\sum_{k\in\mathcal K}n_k$.
    Accordingly, 
    estimators via CVS, pCVS, and O-TL, i.e.,
    \eqref{eq:hat_beta_C}, \eqref{eq:hat_beta_PC}, and \eqref{eq:hat_beta_OTL},
    adjust local estimator $\hat\beta^{(0)}$ in a similar way;
    that is,
    $\hat\beta^{(0)}-\bm\phi^\top\bm c$ with 
    certain $\bm c\in\mathbb R^M$ depending on $\hat{\bm c}^{(k)}$ in \eqref{eq:hat_beta_k}.
    
    Accordingly, 
    this connection holds even for $\hat\beta_{AO}^{(0)}$ in \eqref{eq:hat_beta_AOTL},
    the estimator given by AO-TL, 
    because
    $
        \hat\beta_{AO}^{(0)}=
        a\hat\beta_{O,\mathcal K_{k^*}}^{(0_1)}+(1-a)\hat\beta_{O,\mathcal K_k}^{(0_1)}
    $
    for certain $\mathcal K_k$ in \eqref{eq:K_k} and $a\in[0,1]$.
    Recall that $\hat\beta_{O,\mathcal K_k}^{(0_1)}$ denotes the O-TL estimator
    taking $\mathcal K_k$ as the transferable set and 
    ${\mathcal D}^{(0_1)}$ (a random half of ${\mathcal D}^{(0)}$) as the target data.
    As long as 
    $\hat\beta^{(0_1)}$ (the local estimator utilizing ${\mathcal D}^{(0_1)}$ only)
    well approximates $\hat\beta^{(0)}$
    (arguably correct when $n_0$ is large),
    $\hat\beta_{AO}^{(0)}$ can be viewed as the difference between
    $\hat\beta^{(0)}$ and offset $\bm\phi^\top\bm c$
    with $\bm c$ depending on $\hat{\bm c}^{(k)}$ in \eqref{eq:hat_beta_k}.
\end{remark}

\section{Theoretical results}\label{sec:theory}

For convenience,
we assume 
$n_k$, $J_k$, $\rho^{(k)}$, and $\lambda^{(k)}$ are identical across all datasets
and are hence abbreviated as $n$, $J$, $\rho$, and $\lambda$, respectively;
i.e.,
$n_0=\cdots=n_K=n$,
$J_0=\cdots=J_K=J$,
$\rho^{(0)}=\cdots=\rho^{(K)}=\rho$, and
$\lambda^{(0)}=\cdots=\lambda^{(K)}=\lambda$.
Additionally, 
basis functions $\phi_m(t)$ are specified as
$$  
    \phi_m(t)=
    \begin{cases}
        1, &m=1,\\
        \sqrt{2}\cos(2\pi(m/2)t), &m\geq 2 \text{ and }m \text{ is even},\\
        \sqrt{2}\sin(2\pi\{(m-1)/2\}t), &m\geq 3 \text{ and }m \text{ is odd}.
    \end{cases}
$$
As a result,
$\bm\Psi$ in \eqref{eq:matrix_Psi} reduces to an identity matrix.
Also,
in this section,
we set 
$$
    M=1+2\lfloor(J-1)/2\rfloor\approx J
$$ 
in which $\lfloor a\rfloor$ denoting the largest integer less than or equal to $a$.
$\bm W$ in \eqref{eq:matrix_W} is thus diagonal:
$$
    \bm W=
    16\pi^4{\rm diag}\{0,1^4,1^4,2^4,2^4,\ldots,\lfloor(J-1)/2\rfloor^4,\lfloor(J-1)/2\rfloor^4\}.
$$

Our theoretical results are built upon following conditions.
Among them, 
\eqref{cond:trajectories} through \eqref{cond:lambda_n} 
are borrowed from \cite{LiHsing2007},
ensuring satisfactory recovery of underlying trajectories 
and the convergence of local estimators without information transfer.
The introduction of \eqref{cond:xi} is bifold.
Firstly, it is comparable to Assumption 4 of \cite{TianFeng2023},
restricting the dissimilarity between 
the target and sources.
In the most ideal case 
(when covariance functions $C^{(k)}$ are quite the same across $k$),
the largest eigenvalue of
$\bm\Omega^{(k)-1}\bm\Omega^{(0)}$
is bounded in probability.
On the other hand,
\eqref{cond:xi} expects a larger value of $n$ than that in \eqref{cond:lambda_n}
to accommodate the potential difference across multiple covariance functions.
The remaining  \eqref{cond:zeta},
specific for the convergence of pCVS,
covers the requirement on the additional tuning parameter $\zeta$.
In particular,
if $J^\xi=O(1)$ and $\lambda\sim n^{-3}$,
then the conditions outlined in \citet[Theorems 1 \& 2]{DingLiXieWang2024}
constitute a special case of \eqref{cond:zeta}.

\begin{enumerate}[label=C\arabic*]
    \item\label{cond:trajectories}
        For each $k\in\{0,\ldots,K\}$,
        with probability one,
        $X_i^{(k)}$ belongs to the periodic Sobolev space
        \begin{multline*}
            W_{2,{\rm per}}^2=
            \{f\in L^2[0,1]: 
            f\text{ is twice differentiable where $f'$ is absolutely continuous}
            \\
            \text{ with $f''\in L^2[0,1]$, $f(0)=f(1)$ and $f'(0)=f'(1)$}\};
        \end{multline*}
        $\sup_{t\in[0,1]}{\rm E}|{\rm d}^2X_1^{(k)}(t)/{\rm d}t^2|^2<\infty$;
        ${\rm E}\|{\rm d}^2X_1^{(k)}(t)/{\rm d}t^2\|_{L^2}^4<\infty$;
        ${\rm E}|\epsilon_{ij}^{(k)}|^4<\infty$.
    \item\label{cond:rho_J} 
        $\rho\to0$ and $J^{-1}\rho^{-1/4}\to0$.
    \item\label{cond:lambda_n}
        $\lambda\to0$ 
        and
        $n^{-1}\lambda^{-1/4}\to0$.
    \item\label{cond:xi}
        There is $\xi\geq 0$ such that
        \begin{enumerate}[label=\theenumi.\arabic*]
            \item\label{cond:cov_ratio}
                the largest eigenvalue of $M\times M$ matrix $\bm\Omega^{(k)-1}\bm\Omega^{(0)}$
                is of order $J^\xi$ ($=M^\xi$) in probability,
                i.e.,
                $\|\bm\Omega^{(k)-1}\bm\Omega^{(0)}\|_2
                =O_p(J^\xi)$;
            \item
                $\lambda\to 0$ and
                $n^{-1}\lambda^{-1/4}J^\xi\to0$.
        \end{enumerate}
    \item\label{cond:zeta}
        $n^{-2}\lambda^{-1}\zeta^2J^\xi\to 0$
        with $\xi\geq 0$ satisfying \eqref{cond:cov_ratio}.
\end{enumerate}

Let $\widehat C^{(k)}$ denote the empirical version of covariance function $C^{(k)}$.
Specifically,
$$
    \widehat C^{(k)}(s,t)=
    \frac{1}{n_k}\sum_{i=1}^{n_k}\widehat X_i(s)\widehat X_i(t).
$$
Define $\langle\cdot,\cdot\rangle_{\widehat C^{(k)}}$ and $\|\cdot\|_{\widehat C^{(k)}}$ such that
\begin{equation}\label{eq:C_norm}
    \langle f, g\rangle_{\widehat C^{(k)}}
    =
    \int_0^1\int_0^1f(s)\widehat C^{(k)}(s,t)g(t){\rm d}s{\rm d}t
    \quad\text{and}\quad
    \|f\|_{\widehat C^{(k)}}
    =
    \sqrt{\langle f, f\rangle_{\widehat C^{(k)}}}.
\end{equation}
We first present results on CVS.
Recall that the estimator given by CVS is $\hat\beta_C^{(0)}$ in \eqref{eq:hat_beta_C}
rather than the ideal (but infeasible) $\tilde\beta^{(0)}$ in \eqref{eq:tilde_beta_0}
which shrinks the mean squared error of 
local estimator $\hat\beta^{(0)}$ in \eqref{eq:hat_beta_k} pointwisely.
Proposition \ref{prop:hat_beta_c_estimation_err}
justifies $\hat\beta_C^{(0)}$ by showing that
$\hat\beta_C^{(0)}$ converges to $\tilde\beta^{(0)}$
in $\|\cdot\|_{\widehat C^{(0)}}$-norm.
Also,
it reveals
the rate for $\hat\beta_C^{(0)}$ converging to the true $\beta^{(0)}$.
This rate is comparable to Theorem 3 in \cite{LiHsing2007},
which considers the target data only.
The additional term $J^\xi$ captures the influence of source covariance functions $C^{(k)}$, $k=1,\ldots,K$: 
when the covariance functions of the source data are close to that of the target data, 
the convergence rate is improved.

\begin{proposition}\label{prop:hat_beta_c_estimation_err}
    Under conditions 
    \eqref{cond:trajectories},
    \eqref{cond:rho_J},
    and \eqref{cond:xi},
    $$
        {\rm E}(\|\hat\beta_C^{(0)}-\tilde\beta^{(0)}\|_{\widehat C^{(0)}}^2\mid\mathcal Z)
        =O_p(n^{-1}\lambda^{-1/4}J^\xi),
    $$
    and
    $$
        {\rm E}(\|\hat\beta_C^{(0)}-\beta^{(0)}\|_{\widehat C^{(0)}}^2\mid\mathcal Z)
        =
        O_p(\lambda+\rho+J^{-1}\rho^{-1/4}+n^{-1}\lambda^{-1/4}J^\xi).
    $$
\end{proposition}

In terms of the prediction accuracy,
$\hat\beta_C^{(0)}$ is comparable to $\tilde\beta^{(0)}$ too.
Specifically,
suppose there is a testing set (indexed by $(00)$)
consisting of $n$ new independent curves,
say $X_1^{(00)},\ldots,X_n^{(00)}$,
from the target population 
but independent of $X_1^{(0)},\ldots,X_n^{(0)}$.
The difference in prediction 
between $\hat\beta_C^{(0)}$ and $\tilde\beta^{(0)}$
is convergent
when the prediction is made with $\widehat X_i^{(00)}$,
which are smoothed counterparts of $X_i^{(00)}$ and
recovered from contaminated observations $\bm Z^{(00)}$.
Even if the predicted values are compared with  the true ones,
the error is guaranteed to be convergent;
see Proposition \ref{prop:hat_beta_c_pred_err}.
\begin{proposition}\label{prop:hat_beta_c_pred_err}
    With the conditions of Proposition \ref{prop:hat_beta_c_estimation_err},
    $$
    {\rm E}\left(
    \frac{1}{n}\sum_{i=1}^n
    \langle\tilde\beta^{(0)}-\hat\beta_C^{(0)},\widehat X_i^{(00)}\rangle_{L^2}^2
    \mid\mathcal Z,\bm Z^{(00)}
    \right)
    =O_p(n^{-1}\lambda^{-1/4}J^\xi),
    $$
    and
    $$
    {\rm E}\left(
    \frac{1}{n}\sum_{i=1}^n
    \left\{
    \langle\beta^{(0)},X_i^{(00)}\rangle_{L^2}-
    \langle\hat\beta_C^{(0)},\widehat X_i^{(00)}\rangle_{L^2}
    \right\}^2
    \mid\mathcal Z,\bm Z^{(00)}
    \right)
    =
    O_p(\lambda+\rho+J^{-1}\rho^{-1/4}+n^{-1}\lambda^{-1/4}J^\xi).
    $$
\end{proposition}

Proposition \ref{prop:hat_beta_pc_err} showcases that
the convergence rates of pCVS are comparable to those established 
in Propositions \ref{prop:hat_beta_c_estimation_err} and \ref{prop:hat_beta_c_pred_err}, 
with one more term accounting for the extra tuning parameter $\zeta$.

\begin{proposition}\label{prop:hat_beta_pc_err}
    Holding conditions     
    \eqref{cond:trajectories},
    \eqref{cond:rho_J},
    \eqref{cond:xi},
    and \eqref{cond:zeta},
    $$
        {\rm E}(\|\hat\beta_{PC}^{(0)}-\tilde\beta^{(0)}\|_{\widehat C^{(0)}}^2\mid\mathcal Z)
        =
        O_p(n^{-1}\lambda^{-1/4}J^\xi+n^{-2}\lambda^{-1}\zeta^2J^\xi),
    $$
    $$
        {\rm E}(\|\hat\beta_{PC}^{(0)}-\beta^{(0)}\|_{\widehat C^{(0)}}^2\mid\mathcal Z)
        =O_p(\lambda+\rho+J^{-1}\rho^{-1/4}+n^{-1}\lambda^{-1/4}J^\xi+n^{-2}\lambda^{-1}\zeta^2J^\xi),
    $$
    $$
    {\rm E}\left(
    \frac{1}{n}\sum_{i=1}^n
    \langle\tilde\beta^{(0)}-\hat\beta_{PC}^{(0)},\widehat X_i^{(00)}\rangle_{L^2}^2
    \mid\mathcal Z,\bm Z^{(00)}
    \right)
    =O_p(n^{-1}\lambda^{-1/4}J^\xi+n^{-2}\lambda^{-1}\zeta^2J^\xi),
    $$
    and
    \begin{multline*}
    {\rm E}\left(
    \frac{1}{n}\sum_{i=1}^n
    \left\{
    \langle\beta^{(0)},X_i^{(00)}\rangle_{L^2}-
    \langle\hat\beta_{PC}^{(0)},\widehat X_i^{(00)}\rangle_{L^2}
    \right\}^2
    \mid\mathcal Z,\bm Z^{(00)}
    \right)
    \\
    =
    O_p(\lambda+\rho+J^{-1}\rho^{-1/4}+n^{-1}\lambda^{-1/4}J^\xi+n^{-2}\lambda^{-1}\zeta^2J^\xi).
    \end{multline*}
\end{proposition}

\section{Numerical illustration}\label{sec:numerical}

In this section, 
we numerically compare estimators produced by O-TL, AO-TL, CVS, and pCVS
with the target-only local estimator, 
focusing on both estimation accuracy and predictive performance.
For a given estimator $\tilde\beta^{(0)}$ of target coefficient $\beta^{(0)}$, 
predictive performance is assessed using the relative prediction error (RPE),
defined as the ratio of the sum of squared prediction errors obtained by 
$\tilde\beta^{(0)}$ to that of the target-only local estimator $\hat\beta^{(0)}$:
$$
	{\rm RPE}(\tilde\beta^{(0)})=
	\frac{\sum_{i\in\text{testing set}}(Y_i^{(0)}-\langle\widehat X_i^{(0)},\tilde\beta^{(0)}\rangle_{L^2})^2}
	{\sum_{i\in\text{testing set}}(Y_i^{(0)}-\langle\widehat X_i^{(0)},\hat\beta^{(0)}\rangle_{L^2})^2}.
$$
When $\beta^{(0)}$ is known,
we further evaluate estimation accuracy via the relative estimation error (REE) 
in $\|\cdot\|_{\widehat C^{(0)}}$-norm:
$$
	{\rm REE}(\tilde\beta^{(0)})=
	\frac{\|\tilde\beta^{(0)}-\beta^{(0)}\|_{\widehat C^{(0)}}^2}
	{\|\hat\beta^{(0)}-\beta^{(0)}\|_{\widehat C^{(0)}}^2},
$$
with $\|\cdot\|_{\widehat C^{(0)}}$ defined in \eqref{eq:C_norm}.

\subsection{Simulation study}

We generate 100 simulated datasets.
In each dataset, 
there is one target dataset and four source datasets, 
each consisting of 300 independent subjects.
The underlying functional predictors $X_i^{(k)}$ 
are assumed to be zero-mean Gaussian processes with covariance function
$$
	C^{(k)}(s,t)=
	\begin{cases}
		10\times\exp(-15|s-t|)
		&\text{if $k=0$}
		\\
		\eta \exp(-15|s-t|)
		&\text{if $k\in\{1,\ldots,K\}$},
	\end{cases}
$$
where $\eta\in\{100, 50, 10, 5, 1\}$.
Contaminated observations $Z_{i,j}^{(k)}$ are collected at 50 evenly spaced spots on $[0,1]$.
The target and source models share the identical coefficient function
$$
	\beta^{(k)}(t)=P_1(t)+P_2(t)
$$
in which $P_1$ and $P_2$ are (normalized) shifted Legendre polynomials 
of orders 1 and 2, respectively;
see \citet[pp.~773--774]{Hochstrasser1972}.
Generated using the \texttt{R} package \texttt{orthopolynom} \citep{R-orthopolynom},
polynomials $P_1$ and $P_2$ are unit-normed and mutually orthogonal on $[0,1]$.
Error terms $\epsilon_{i,j}^{(k)}$ and $\varepsilon_i^{(k)}$ 
are normally distributed with mean zero and a small variance of .01.
Responses $Y_i^{(k)}$ are then generated with zero mean.
For the target data, 
20\% of the subjects are randomly reserved for testing, 
and the remaining 80\% are used for training. 
RPE values are computed on the testing set,
while REEs are checked using the training set.

\begin{figure}[!ht]
    \centering
    \begin{subfigure}{0.48\textwidth}
        \includegraphics[width=\linewidth]{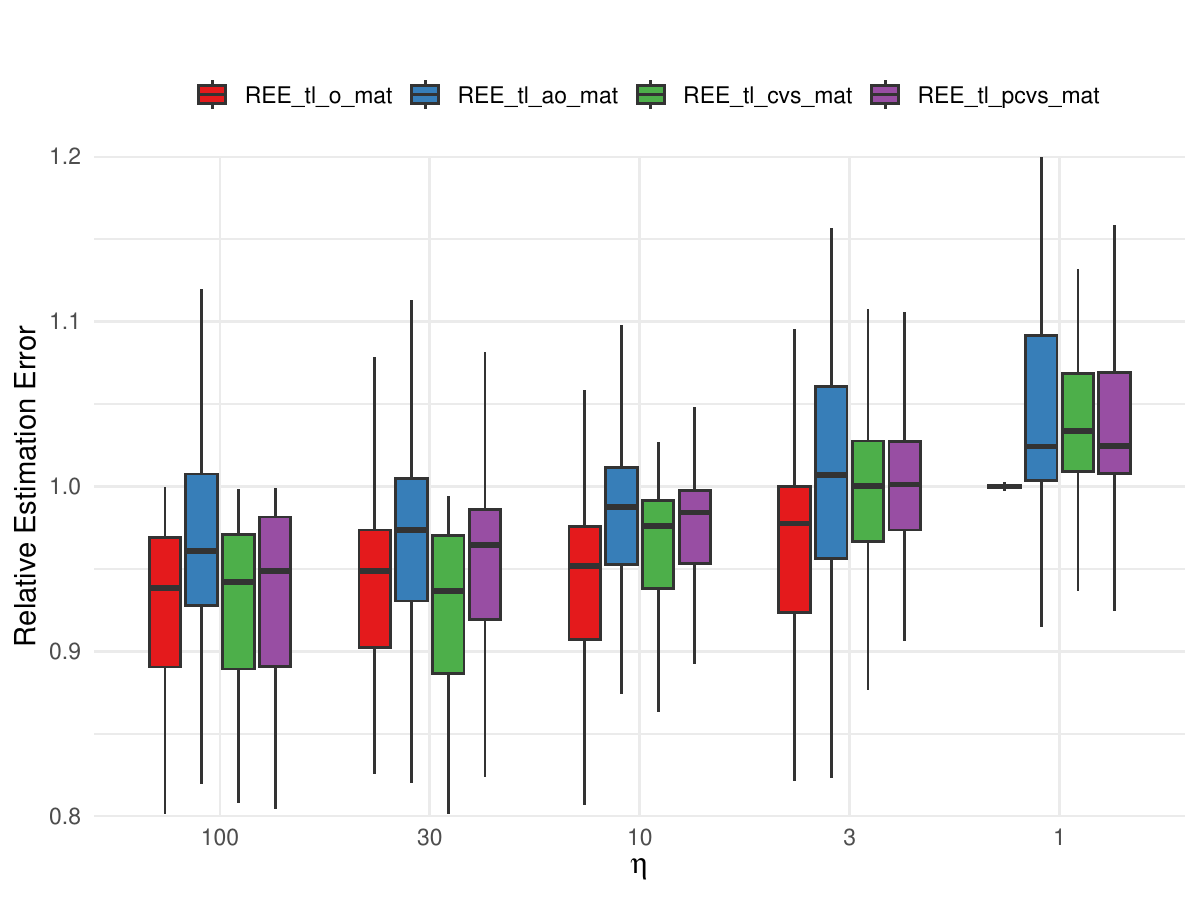}
        \caption{\scriptsize REE values versus $\eta$}
        \label{fig:simu_ree}
    \end{subfigure}\hfill
    \begin{subfigure}{0.48\textwidth}
        \includegraphics[width=\linewidth]{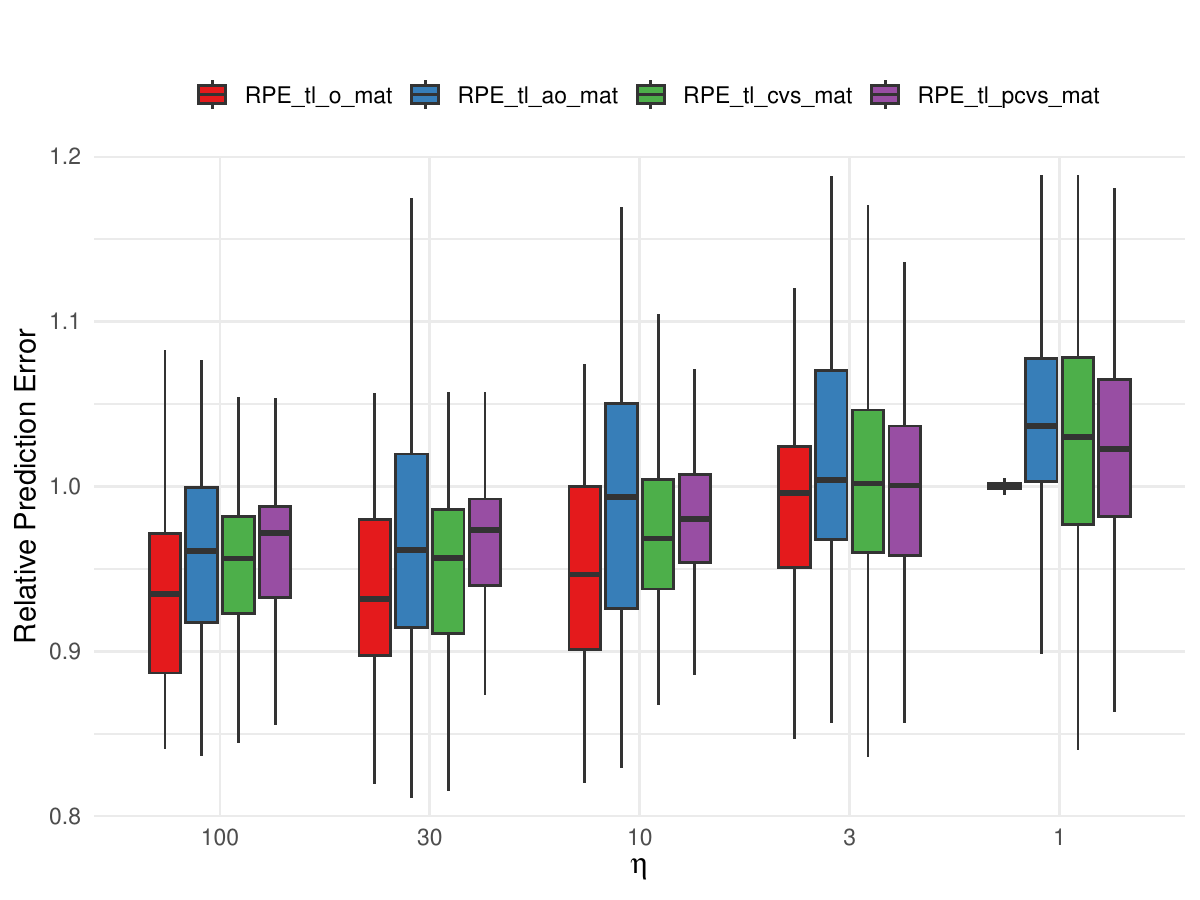}
        \caption{\scriptsize RPE value sversus $\eta$}
        \label{fig:simu_rpe}
    \end{subfigure}
    
    \caption{
    	\small
		Boxplots of values of REE (left panel) and RPE (right panel) for O-TL, AO-TL, CVS, and pCVS based on 100 simulative datasets.
		In each subfigure,
		boxplots are generated for different values of $\eta$.
    }
    \label{fig:simu}
\end{figure}

As shown in Figure \ref{fig:simu},
O-TL achieves the best performance for almost all values of $\eta$,
which is expected since all source datasets share the same distribution and, 
moreover, the coefficient functions are identical across the target and sources.
The estimation and prediction accuracy of CVS and pCVS are comparable to that of O-TL under this setting.
By contrast, 
AO-TL performs worse than O-TL, 
as it does not necessarily incorporate all available source datasets,
even though all sources are transferable in this setting.
In addition, 
Figure \ref{fig:simu} 
indicates that the performance of CVS and pCVS deteriorates as $\eta$ decreases. 
This observation supports
the convergence rates derived in Section~\ref{sec:theory},
which depend on $\|\bm\Omega^{(k)-1}\bm\Omega^{(0)}\|_2$;
under this simulative setup, $\|\bm\Omega^{(k)-1}\bm\Omega^{(0)}\|_2\approx 10/\eta$.

\subsection{Application to the prediction of stock returns}

The objective of portfolio management is to balance expected returns and risk, 
which motivates investors to periodically rebalance their portfolios. 
Consequently, some investors are interested in predicting future stock returns within one specific sector. 
TF can facilitate this task by borrowing information from related sectors in the broader market.
Specifically, 
denote by $s_i^{(k)}(t_{j,\ell})$ the daily closing prices of 
the $i$th stock in the $k$th section at the $j$th day of the $\ell$th month.
Fixing $\ell$,
SoFR may help predict the monthly (close-to-close) return (MR) in the $(\ell+1)$th month
$$
	Y_i^{(k,\ell+1)}=\frac{s_i^{(k)}(t_{J,\ell+1})-s_i^{(k)}(t_{1,\ell+1})}{s^{(k)}(t_{1,\ell+1})},
$$
utilizing the monthly cumulative (close-to-close) return (MCR) from the $\ell$th month, i.e.,
$$
	Z_{i,j}^{(k,\ell)}=\frac{s_i^{(k)}(t_{j,\ell})-s_i^{(k)}(t_{1,\ell})}{s^{(k)}(t_{1,\ell})}.
$$

We use the dataset preprocessed by \cite{LinReimherr2024} 
(available at \url{github.com/haotianlin/HTL-FLM/}).
Spanning from May 1, 2021 to September 30, 2021,
it consists of 11 sectors of 
Nasdaq-listed stocks with market capitalizations exceeding 20 billion USD: 
basic industries (BI), capital goods (CG),
consumer durable (CD), consumer non-durable (CND), consumer services (CS), energy (E), 
finance (Fin), health care (HC), public utility (PU), technology (Tech), and transportation (Trans). 
The numbers of stocks in these sectors are 60, 58, 31, 30, 104, 55, 70, 68, 46, 103, and 41, respectively. 
When a single sector's data are used for model fitting, 
the small sample size may limit the performance.

We repeat the following experiment 100 times.
In each repetition,
we cycle through the 11 sectors, 
treating one sector as the target and the remaining sectors as sources. 
For a fixed target sector, 
we randomly split its data into training and testing sets (80/20) and 
compute RPE values for O-TL 
(taking all the source sectors as transferable), 
AO-TL, CVS, and pCVS. 
The resulting RPE values are presented in Figure \ref{fig:stocks}.

\begin{figure}[!ht]
    \centering
    \begin{subfigure}{0.48\textwidth}
        \includegraphics[width=\linewidth]{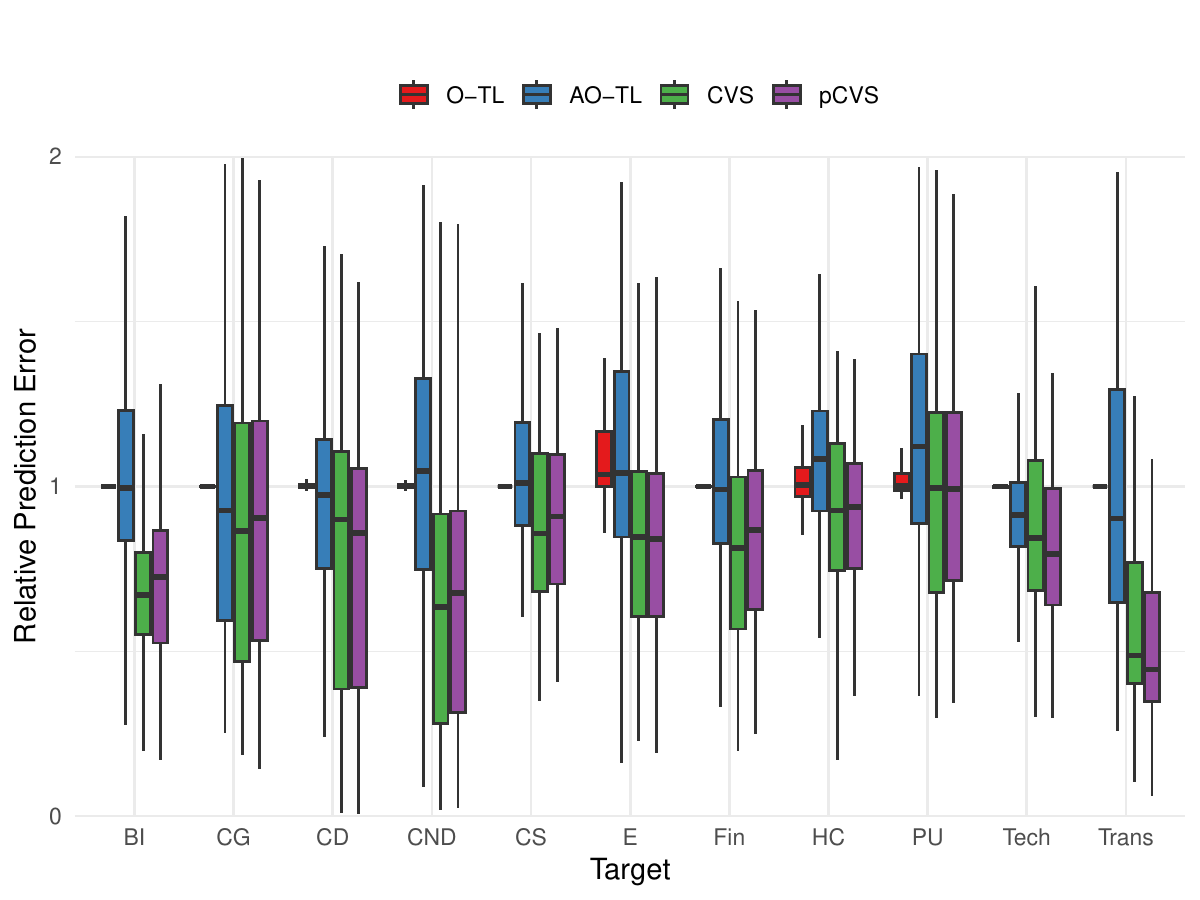}
        \caption{\scriptsize Predicting MRs in June based on MCRs in May}
        \label{fig:stocks56}
    \end{subfigure}\hfill
    \begin{subfigure}{0.48\textwidth}
        \includegraphics[width=\linewidth]{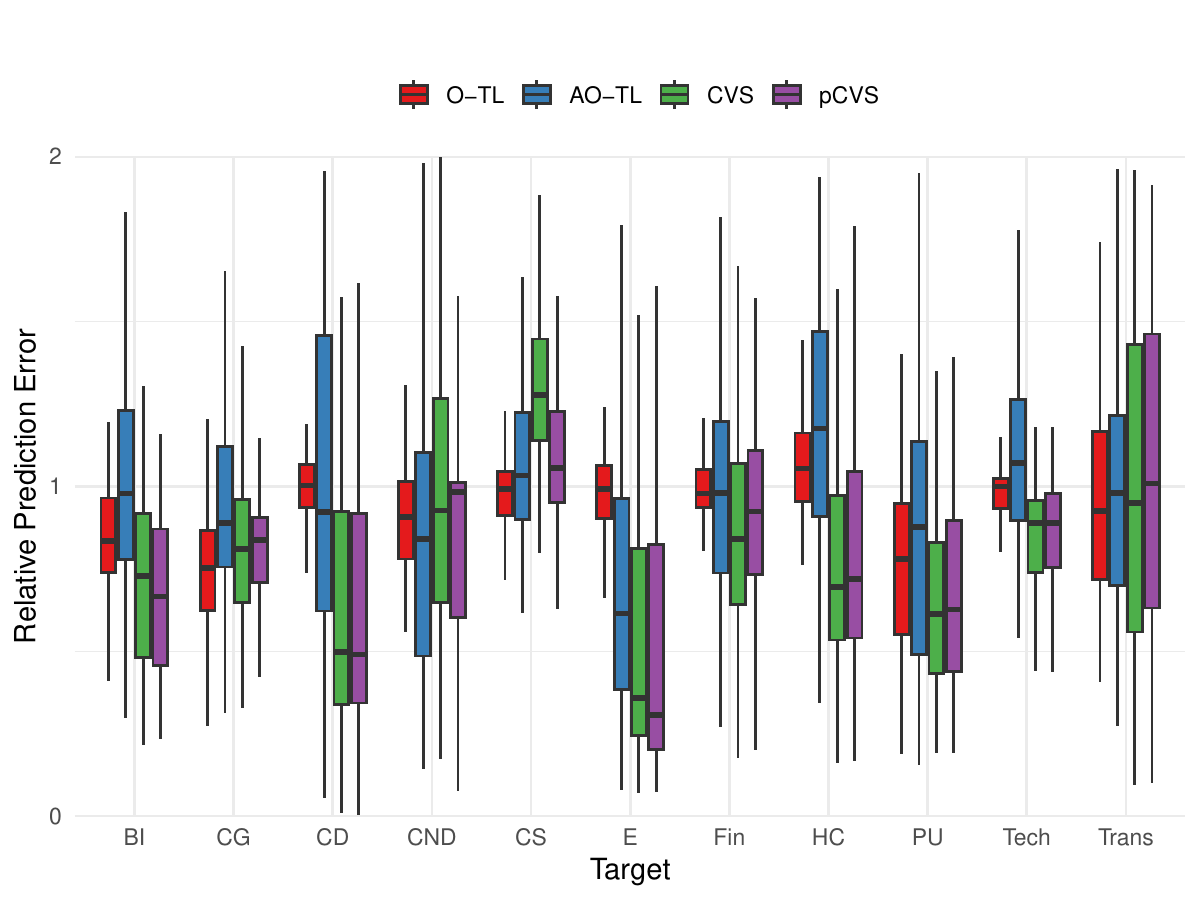}
        \caption{\scriptsize Predicting MRs in July based on MCRs in June}
        \label{fig:1b}
    \end{subfigure}
    
    \begin{subfigure}{0.48\textwidth}
        \includegraphics[width=\linewidth]{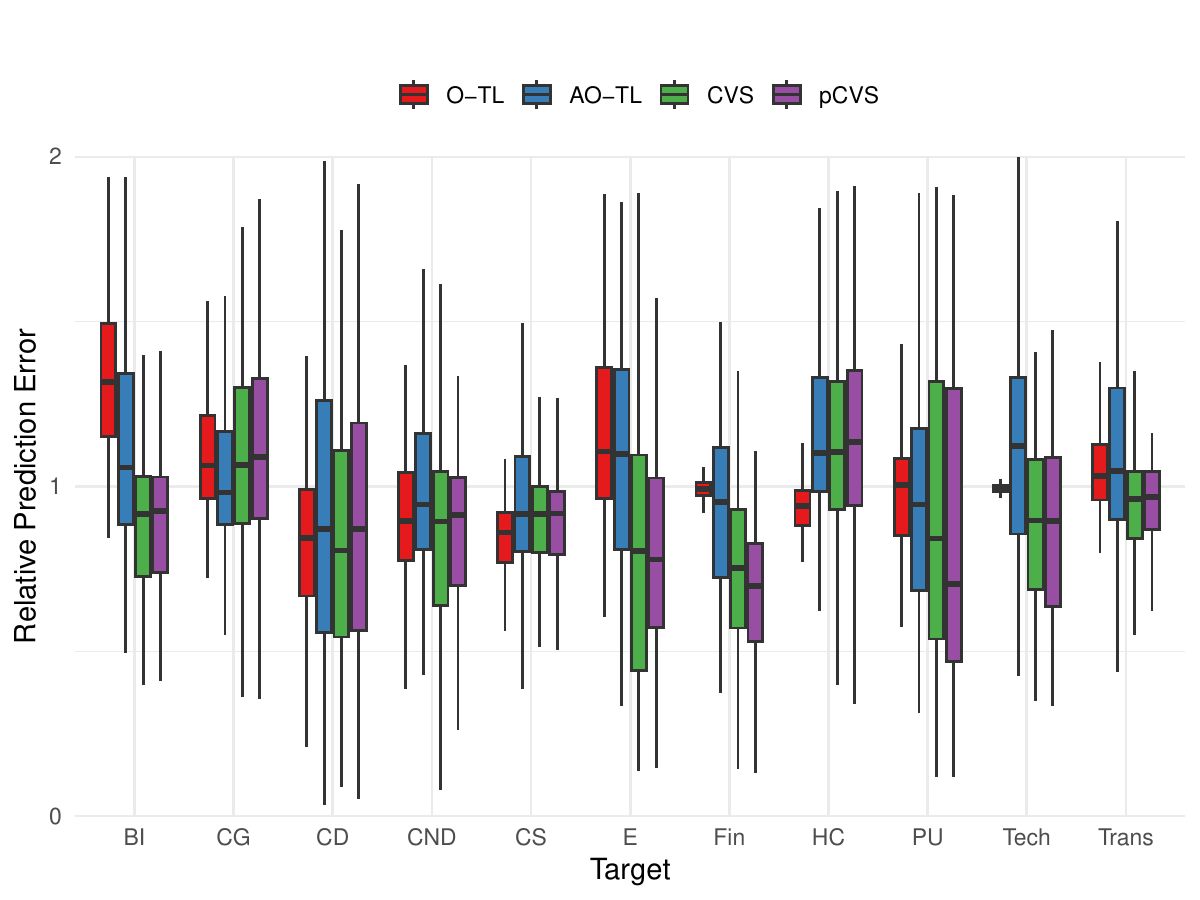}
        \caption{\scriptsize Predicting MRs in August based on MCRs in July}
        \label{fig:1c}
    \end{subfigure}\hfill
    \begin{subfigure}{0.48\textwidth}
        \includegraphics[width=\linewidth]{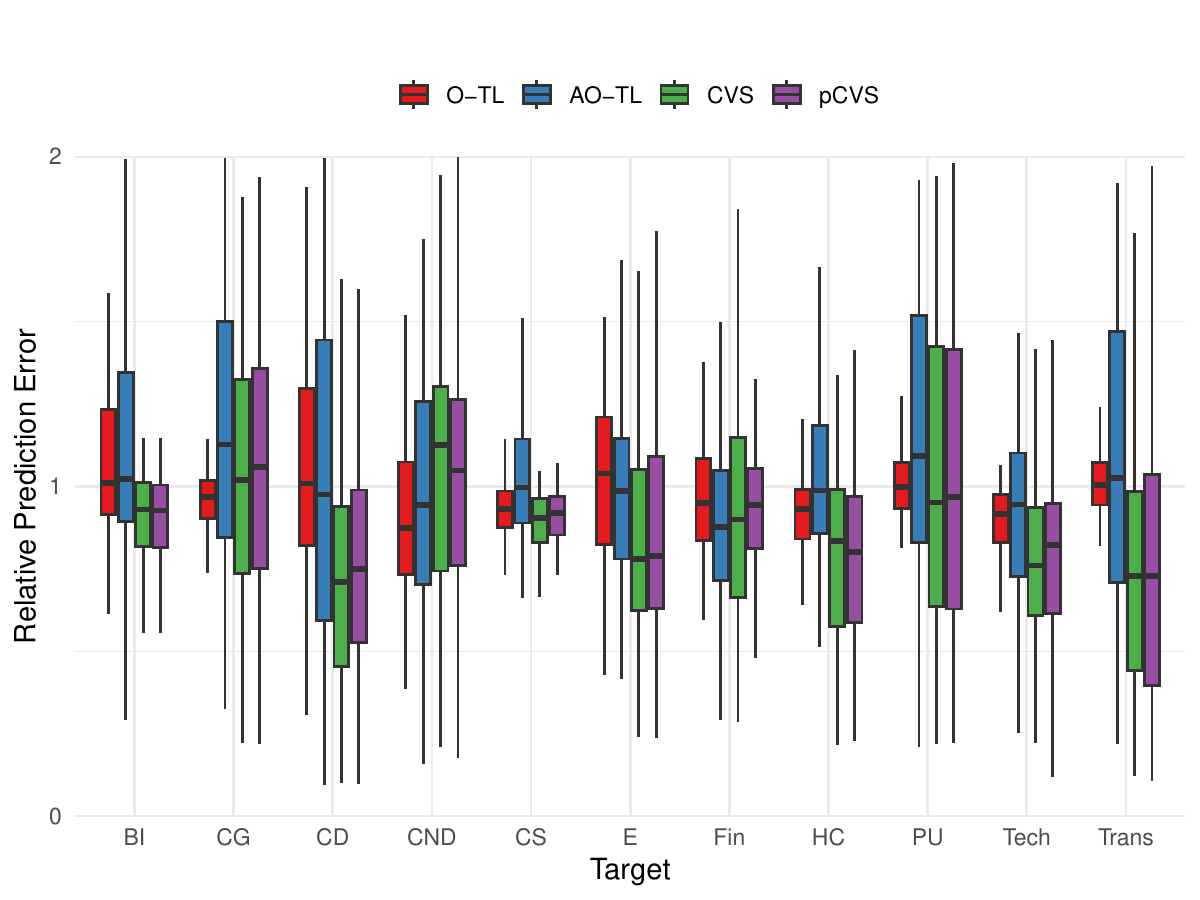}
        \caption{\scriptsize Predicting MRs in September based on MCRs in August}
        \label{fig:1d}
    \end{subfigure}
    
    \caption{
    	\small
		Boxplots of RPE values for O-TL, AO-TL, CVS, and pCVS based on 100 replications.
		Each subfigure corresponds to the prediction of MRs in a given month using MCRs from the preceding month.
		For each month pair, 
		we cycle through the 11 sectors: 
		basic industries (BI), capital goods (CG),
		consumer durable (CD), consumer non-durable (CND), consumer services (CS), energy (E), 
		finance (Fin), health care (HC), public utility (PU), technology (Tech), and transportation (Trans).
		Values of RPE are computed by
		treating each sector in turn as the target and the remaining sectors as sources.
    }
    \label{fig:stocks}
\end{figure}

We observe that O-TL (treating all source sectors as transferable) 
reduces prediction error only in a limited number of scenarios, 
while in most cases it yields little improvement or even degrades predictive performance; 
see in particular Figure \ref{fig:stocks56}.
In the absence of prior knowledge about the transferable set, 
O-TL can behave erratically: 
it improves prediction when the target sector shares strong similarities with the sources, 
but performs poorly when such similarities are weak.
By contrast, 
AO-TL exhibits overall better performance than O-TL, 
highlighting the importance of identifying the transferable set prior to aggregation. 
However, because this identification relies on predictive performance evaluated on validation data, 
its effectiveness is constrained by the limited sample size here.
Compared with AO-TL, 
CVS and pCVS achieve RPE values of similar magnitude across almost all scenarios, 
although their predictive performance exhibits greater variability. 
We conjecture that this increased variability arises 
because both CVS and pCVS depend on estimating the variances of local estimators, 
which is also sensitive to small sample sizes.

\section{Conclusion and discussion}\label{sec:conclusion}

In this work, 
we repurpose the CVS method for TL in the context of SoFR. 
Our results demonstrate that 
CVS-based TL provides an effective alternative to O-TL and its variants, 
particularly in settings where subject-level data cannot be shared. 
Beyond methodological development, 
our analysis reveals a deeper connection between these two TL strategies. 
Although they originate from different motivations, 
we show that they are closely related in how they enhance local estimators, 
offering new theoretical insight into the structure of TL procedures. 
In addition, our theoretical results underscore the importance of accounting for smoothing error and highlight the critical role played by the covariance similarity across datasets. 

While the proposed estimators perform competitively in numerical studies, 
several challenges remain. 
In particular, 
accurately estimating and inverting the covariance structure of local estimators 
continues to be a key practical bottleneck, 
especially in small-sample regimes. 
Addressing this issue through improved variance estimation techniques or alternative regularization strategies represents an important direction for future research.

Compared with O-TL and AO-TL, 
CVS-based TL has received relatively less attention in practice, 
partly due to its less straightforward implementation,
even though it offers distinct advantages in privacy-sensitive settings.
More broadly, 
the CVS framework is not limited to SoFR or even FDA.
As long as covariance structures of local estimators can be reliably estimated, 
CVS-based TL may serve as a unifying framework for developing privacy-aware and theoretically grounded TL-enhanced estimators across a wide range of parametric models.

\bibliographystyle{apalike}
\bibliography{mybib}

@article{Bastani2021,
    author = {Hamsa Bastani},
    title = {Predicting with Proxies: Transfer Learning in High Dimension},
    volume = {67},
    journal = {Management Science},
    pages = {2964--2984},
    year = {2021},
    doi = {10.1287/mnsc.2020.3729},
}

@article{DingLiXieWang2024,
    author = {Ding, Jie and Li, Jialiang and Xie, Ping and Wang, Xiaoguang},
    title = {Efficient Risk Assessment of Time-to-Event Targets With Adaptive Information Transfer},
    journal = {Statistics in Medicine},
    volume = {43},
    pages = {6026--6041},
    doi = {10.1002/sim.10290},
    year = {2024}
}

@article{GaiffasLecue2011,
    author  = {St{{\'e}}phane Ga{{\^i}}ffas and Guillaume Lecu{{\'e}}},
    title   = {Hyper-Sparse Optimal Aggregation},
    journal = {Journal of Machine Learning Research},
    year    = {2011},
    volume  = {12},
    pages   = {1813--1833},
    url     = {http://jmlr.org/papers/v12/gaiffas11a.html}
}

@incollection{Hochstrasser1972,
    author = {Urs W. Hochstrasser},
    title = {Orthogonal Polynomials},
    editor = {Milton Abramowitz and Irene A. Stegun},
    booktitle = {Handbook of Mathematical Functions with Formulas, Graphs, and Mathematical Tables},
    year = {1972},
    series = {Applied Mathematics Series 55},
    publisher = {Dover Publications, Inc.},
    address = {New York},
    note = {Tenth original printing with corrections},
    pages = {773--802},
}

@Article{LiangEtal2024,
  title = {{sparsegl}: An {R} Package for Estimating Sparse Group Lasso},
  author = {Xiaoxuan Liang and Aaron Cohen and Anibal {Sol{\'o}n Heinsfeld} and Franco Pestilli and Daniel J. McDonald},
  journal = {Journal of Statistical Software},
  year = {2024},
  volume = {110},
  pages = {1--23},
  doi = {10.18637/jss.v110.i06},
}

@article{LiCaiLi2022a,
    author = {Li, Sai and Cai, T. Tony and Li, Hongzhe},
    title = {Transfer Learning for High-Dimensional Linear Regression: Prediction, Estimation and Minimax Optimality},
    journal = {Journal of the Royal Statistical Society Series B (Statistical Methodology)},
    volume = {84},
    pages = {149--173},
    year = {2022},
    doi = {10.1111/rssb.12479},
}

@article{LiHsing2007,
    title = {On rates of convergence in functional linear regression},
    journal = {Journal of Multivariate Analysis},
    volume = {98},
    pages = {1782--1804},
    year = {2007},
    doi = {https://doi.org/10.1016/j.jmva.2006.10.004},
    author = {Yehua Li and Tailen Hsing},
}

@InProceedings{LinReimherr2024,
    title =  {On Hypothesis Transfer Learning of Functional Linear Models},
    author = {Lin, Haotian and Reimherr, Matthew},
    booktitle = {Proceedings of the 41st International Conference on Machine Learning},
    year = {2024},
    pages = {30252--30285},
    volume = {235},
    series = {Proceedings of Machine Learning Research},
    editor = {Salakhutdinov, Ruslan and Kolter, Zico and Heller, Katherine and Weller, Adrian and Oliver, Nuria and Scarlett, Jonathan and Berkenkamp, Felix},
}

@article{LiShenNing2023,
    author = {Ziyi Li and Yu Shen and Jing Ning},
    title = {Accommodating Time-Varying Heterogeneity in Risk Estimation under the {Cox} Model: A Transfer Learning Approach},
    journal = {Journal of the American Statistical Association},
    volume = {118},
    pages = {2276--2287},
    year = {2023},
    doi = {10.1080/01621459.2023.2210336},
}

@article{MacKinnonWhite1985,
  title = {Some heteroskedasticity-consistent covariance matrix estimators with improved finite sample properties},
  journal = {Journal of Econometrics},
  volume = {29},
  pages = {305--325},
  year = {1985},
  doi = {10.1016/0304-4076(85)90158-7},
  author = {James G MacKinnon and Halbert White},
}

@book{MarshallOlkinArnold2011,
    title = {Inequalities: Theory of Majorization and Its Applications},
    author = {Albert W. Marshall and Ingram Olkin and Barry C. Arnold},
    series = {Springer Series in Statistics},
    year = {2011},
    Edition = {2nd},
    Publisher = {Springer},
    Address = {New York},
    doi = {10.1007/978-0-387-68276-1},
}

@article{MerikoskiKumar2004,
    author = {Jorma K. Merikoski and Ravinder Kumar},
    title = {Inequalities For Spreads Of Matrix Sums And Products},
    journal = {Applied Mathematics E-Notes},
    volume = {4},
    pages = {150--159},
    year = {2004},
}

@Manual{R-orthopolynom,
    title = {orthopolynom: Collection of Functions for Orthogonal and Orthonormal Polynomials},
    author = {Frederick Novomestky},
    year = {2022},
    note = {R package version 1.0-6.1},
    url = {https://CRAN.R-project.org/package=orthopolynom},
 }

@ARTICLE{PanYang2009,
  author={Pan, Sinno Jialin and Yang, Qiang},
  journal={IEEE Transactions on Knowledge and Data Engineering}, 
  title={A Survey on Transfer Learning}, 
  year={2009},
  volume={22},
  pages={1345--1359},
  doi={10.1109/TKDE.2009.191}}

@article{PeiYuShen2025,
    author = {Yan-Bo Pei and Zheng-Yang Yu and Jun-Shan Shen},
    journal = {BMC Bioinformatics},
    pages = {84},
    title = {Transfer learning for accelerated failure time model with microarray data},
    volume = {26},
    year = {2025}
}

@article{Pruitt1966,
    author = {William E. Pruitt},
    journal = {Journal of Mathematics and Mechanics},
    pages = {769--776},
    title = {Summability of Independent Random Variables},
    volume = {15},
    year = {1966}
}

@book{RamsaySilverman2005,
	Author = {Ramsay, J. O. and Silverman, B. W.},
	Title = {Functional Data Analysis},
	Edition = {2nd},
	Publisher = {Springer},
	Address = {New York},
	Year = {2005},
	doi = {10.1007/b98888},
}

@book{Schott2017,
    title = {Matrix Analysis for Statistics},
    author = {Schott, James R.},
    series = {Wiley Series in Probability and Statistics},
    year = {2017},
	Edition = {3rd},
	Publisher = {John Wiley \& Sons},
	Address = {Hoboken},
}

@article{TianFeng2023,
    author = {Ye Tian and Yang Feng},
    title = {Transfer Learning Under High-Dimensional Generalized Linear Models},
    journal = {Journal of the American Statistical Association},
    volume = {118},
    pages = {2684--2697},
    year = {2023},
    doi = {10.1080/01621459.2022.2071278},
}

@incollection{TorreyShavlik2009,
	author	= "Lisa Torrey and Jude Shavlik",
	title	= "Transfer learning",
	booktitle= "Handbook of Research on Machine Learning Applications and Trends: Algorithms, Methods, and Techniques",
	publisher= "Information Science Reference",
	editor	= "Emilio Soria Olivas and Jose David Martin Guerrero and Marcelino Martinez Sober and Jose Rafael Magdalena Benedito and Antonio Jose Serrano Lopez",
	chapter= "11",
	address= "Hershey ",
	year	= "2009",
}

@article{WeissKhoshgoftaarWang2016,
    author = {Kall Weiss and Taghi M. Khoshgoftaar and DingDing Wang},
    title = {A survey of transfer learning},
    journal = {Journal of Big Data},
    volume = {3},
    pages = {9},
    year = {2016},
    doi = {10.1186/s40537-016-0043-6},
}

@Manual{YangZouBhatnagar2024,
    title = {gglasso: Group Lasso Penalized Learning Using a Unified BMD Algorithm},
    author = {Yi Yang and Hui Zou and Sahir Bhatnagar},
    year = {2025},
    note = {R package version 1.6},
    url = {https://CRAN.R-project.org/package=gglasso},
}

@article{HeLiHuLiu2022,
    author = {He, Yong and Li, Qiushi and Hu, Qinqin and Liu, Lei},
    title = {Transfer learning in high-dimensional semiparametric graphical models with application to brain connectivity analysis},
    journal = {Statistics in Medicine},
    volume = {41},
    pages = {4112--4129},
    doi = {10.1002/sim.9499},
    year = {2022}
}

@article{YuanLin2006,
    author = {Yuan, Ming and Lin, Yi},
    title = {Model selection and estimation in regression with grouped variables},
    journal = {Journal of the Royal Statistical Society: Series B (Statistical Methodology)},
    volume = {68},
    pages = {49--67},
    doi = {10.1111/j.1467-9868.2005.00532.x},
    year = {2006}
}

@article{ZhangDuchiWainwright2015,
  author  = {Yuchen Zhang and John Duchi and Martin Wainwright},
  title   = {Divide and Conquer Kernel Ridge Regression: A Distributed Algorithm with Minimax Optimal Rates},
  journal = {Journal of Machine Learning Research},
  year    = {2015},
  volume  = {16},
  pages   = {3299--3340},
  url     = {http://jmlr.org/papers/v16/zhang15d.html}
}

\if0\blind{
\section*{Acknowledgment}
}\fi

\appendix
\setcounter{equation}{0}
\renewcommand{\theequation}{S\arabic{equation}}

\section{Technical details}

\subsection{Deriving Eq. \eqref{eq:U_star}}
\label{sec:derive_U_star}

Look into
$$
    {\rm var}(\hat{\bm\delta}\mid\mathcal Z)=
    (\mathbf 1_K\mathbf 1_K^\top)\otimes{\rm var}(\hat{\bm c}^{(0)}\mid\bm Z^{(0)})+
    {\rm diag}\big\{
        {\rm var}(\hat{\bm c}^{(1)}\mid\bm Z^{(1)}),\ldots,{\rm var}(\hat{\bm c}^{(K)}\mid\bm Z^{(K)})
    \big\}.
$$
By a generalized Woodbury's matrix identity
\citep[Theorem 5.15]{Schott2017},
\begin{align}
    \notag
    {\rm var}^{-1}&(\hat{\bm\delta}\mid\mathcal Z)
    \\ \notag
    &=
    \Big[
    (\mathbf 1_K\mathbf 1_K^\top)\otimes{\rm var}(\hat{\bm c}^{(0)}\mid\bm Z^{(0)})+
    {\rm diag}\big\{
        {\rm var}(\hat{\bm c}^{(1)}\mid\bm Z^{(1)}),\ldots,{\rm var}(\hat{\bm c}^{(K)}\mid\bm Z^{(K)})
    \big\}
    \Big]^{-1}
    \\ \label{eq:inv_var_delta_hat}
    &=\bm B_1-\bm B_2,
\end{align}
where
$$
   \bm B_1={\rm diag}\{
   {\rm var}^{-1}(\hat{\bm c}^{(1)}\mid\bm Z^{(1)}),\ldots,
   {\rm var}^{-1}(\hat{\bm c}^{(K)}\mid\bm Z^{(K)})
   \}
$$
and
\begin{align*}
    \bm B_2
    &=
    [
    {\rm var}^{-1}(\hat{\bm c}^{(1)}\mid\bm Z^{(1)}),
    \ldots,
    {\rm var}^{-1}(\hat{\bm c}^{(K)}\mid\bm Z^{(K)})
    ]^\top
    \\
    &\qquad
    \bigg\{
    \sum_{k=0}^K
    {\rm var}^{-1}(\hat{\bm c}^{(k)}\mid\bm Z^{(k)})
    \bigg\}^{-1}
    [
    {\rm var}^{-1}(\hat{\bm c}^{(1)}\mid\bm Z^{(1)}),
    \ldots,
    {\rm var}^{-1}(\hat{\bm c}^{(K)}\mid\bm Z^{(K)})
    ].
\end{align*}
A simple algebra gives that
\begin{align*}
    \notag
    \bm U^*
    &=
    \{\mathbf 1_K^\top\otimes{\rm var}(\hat{\bm c}^{(0)}\mid\bm Z^{(0)})\}\bm B_1-
    \{\mathbf 1_K^\top\otimes{\rm var}(\hat{\bm c}^{(0)}\mid\bm Z^{(0)})\}\bm B_2
    \\\notag
    &=
    \bigg\{
    \sum_{k=0}^K
    {\rm var}^{-1}(\hat{\bm c}^{(k)}\mid\bm Z^{(k)})
    \bigg\}^{-1}
    [
    {\rm var}^{-1}(\hat{\bm c}^{(1)}\mid\bm Z^{(1)}),
    \ldots,
    {\rm var}^{-1}(\hat{\bm c}^{(K)}\mid\bm Z^{(K)})
    ].
\end{align*}

\subsection{Approximating expectations and variances conditional on contaminated observations}
\label{sec:approx_exp_var}

For sufficiently large $J_k$,
approximate $Y_i^{(k)}$ by the trapezoidal rule:
\begin{align*}
    Y_i^{(k)}
    &=
    \langle X_i^{(k)},\beta^{(k)}\rangle_{L^2}+\varepsilon_i^{(k)}
    \\
    &= 
    J_k^{-1}\sum_{j=1}^{J_k}\beta^{(k)}(t_j^{(k)})(Z_{i,j}^{(k)}-\epsilon_{i,j}^{(k)})
    \\
	&\qquad
	-(2J_k)^{-1}
	\{
		\beta^{(k)}(t_1^{(k)})(Z_{i,1}^{(k)}-\epsilon_{i,1}^{(k)})+
		\beta^{(k)}(t_{J_k}^{(k)})(Z_{i,J_k}^{(k)}-\epsilon_{i,J_k}^{(k)})
	\}+
	\varepsilon_i^{(k)}
    +O_p(J_k^{-2}).
\end{align*}
It follows that
\begin{align*}
    {\rm E}(\bm Y^{(k)}&\mid\bm Z^{(k)})
    \\
    &\approx 
    J_k^{-1}\left[\sum_{j=1}^{J_k}\beta^{(k)}(t_j^{(k)})Z_{1,j}^{(k)},\ldots,
    \sum_{j=1}^{J_k}\beta^{(k)}(t_j^{(k)})Z_{n_k,j}^{(k)}\right]^\top
    \\
   	&\qquad
   	-(2J_k)^{-1}
   	\left\{
   		\beta^{(k)}(t_1^{(k)})\left[Z_{1,1}^{(k)},\ldots,Z_{n_k,1}^{(k)}\right]^\top+
   		\beta^{(k)}(t_{J_k}^{(k)})\left[Z_{1,J_k}^{(k)},\ldots,Z_{n_k,J_k}^{(k)}\right]^\top
   	\right\}
\end{align*}
and 
\begin{align*}
    {\rm var}(\bm Y^{(k)}&\mid\bm Z^{(k)})
    \\
    &\approx 
    \left(
    	\sigma_{\varepsilon^{(k)}}^2+
    	J_k^{-2}\sigma_{\epsilon^{(k)}}^2
    	\left[
	    	\sum_{j=2}^{J_k-1}\{\beta^{(k)}(t_j^{(k)})\}^2+
	    	\frac{1}{4}\{\beta^{(k)}(t_1^{(k)})\}^2+
	    	\frac{1}{4}\{\beta^{(k)}(t_{J_k}^{(k)})\}^2
    	\right]
    \right)
    \bm I_{n_k}.
\end{align*}
Plugging in empirical counterparts,
their estimators are, respectively,
\begin{align*}
    \widehat{\rm E}(\bm Y^{(k)}&\mid\bm Z^{(k)})
    \\
    &=
    J_k^{-1}\sum_{j=1}^{J_k}\hat\beta^{(k)}(t_j^{(k)})\left[Z_{1,j}^{(k)},\ldots,Z_{n_k,j}^{(k)}\right]^\top
    \\
   	&\qquad
   	-(2J_k)^{-1}
   	\left\{
   		\hat\beta^{(k)}(t_1^{(k)})\left[Z_{1,1}^{(k)},\ldots,Z_{n_k,1}^{(k)}\right]^\top+
   		\hat\beta^{(k)}(t_{J_k}^{(k)})\left[Z_{1,J_k}^{(k)},\ldots,Z_{n_k,J_k}^{(k)}\right]^\top
   	\right\}
    \\
    &=
    J_k^{-1}\bm Z^{(k)\top}\bm D\bm\Phi^{(k)}\hat{\bm c}^{(k)}
\end{align*}
and
\begin{align*}
    \widehat{\rm var}(\bm Y^{(k)}&\mid\bm Z^{(k)})
    \\
    &= 
    \left(
    	\sigma_{\varepsilon^{(k)}}^2+
    	J_k^{-2}\sigma_{\epsilon^{(k)}}^2
    	\left[
	    	\sum_{j=2}^{J_k-1}\{\hat\beta^{(k)}(t_j^{(k)})\}^2+
	    	\frac{1}{4}\{\hat\beta^{(k)}(t_1^{(k)})\}^2+
	    	\frac{1}{4}\{\hat\beta^{(k)}(t_{J_k}^{(k)})\}^2
    	\right]
    \right)
    \bm I_{n_k}
    \\
    &=
    (
    	\hat\sigma_{\varepsilon^{(k)}}^2+J_k^{-2}\hat\sigma_{\epsilon^{(k)}}^2
    	\hat{\bm c}^{(k)\top}\bm\Phi^{(k)\top}\bm D^2\bm\Phi^{(k)}\hat{\bm c}^{(k)}
    )
    \bm I_{n_k},
\end{align*}
where $\bm D={\rm diag}\{1/2,1,1,\ldots,1,1,1/2\}\in\mathbb R^{J_k\times J_k}$,
\begin{align*}
    \hat\sigma_{\epsilon^{(k)}}^2
    &=
    (n_kJ_k)^{-1}
    \sum_{i=1}^{n_k}
    (\bm Z_i^{(k)}-\bm\Phi^{(k)}\bm P^{(k)}\bm Z_i^{(k)})^\top
    (\bm Z_i^{(k)}-\bm\Phi^{(k)}\bm P^{(k)}\bm Z_i^{(k)})
    \\
    &=
    (n_kJ_k)^{-1}
    {\rm tr}\{
	    (\bm Z^{(k)}-\bm\Phi^{(k)}\bm P^{(k)}\bm Z^{(k)})^\top
	    (\bm Z^{(k)}-\bm\Phi^{(k)}\bm P^{(k)}\bm Z^{(k)})
    \},
    \\
    \hat\sigma_{\varepsilon^{(k)}}^2
    &=
    n_k^{-1}
    (\bm Y^{(k)}-\bm Z^{(k)\top}\bm P^{(k)\top}\bm\Psi\hat{\bm c}^{(k)})^\top
    (\bm Y^{(k)}-\bm Z^{(k)\top}\bm P^{(k)\top}\bm\Psi\hat{\bm c}^{(k)}),
\end{align*}
and $\bm I$ denotes an identity matrix.
Alternatively,
one may follow \cite{MacKinnonWhite1985}
and make use of
the following estimator of ${\rm var}(\bm Y^{(k)}\mid\bm Z^{(k)})$:
\begin{multline*}
	\widehat{\rm var}(\bm Y^{(k)}\mid\bm Z^{(k)})
	=
	\frac{n_k}{n_k-M}
	{\rm diag}\{
		(Y_1^{(k)}-\widehat{\rm E}(Y_1^{(k)}\mid\bm Z^{(k)}))^2,
		\ldots,
		(Y_{n_k}^{(k)}-\widehat{\rm E}(Y_{n_k}^{(k)}\mid\bm Z^{(k)}))^2
	\},
\end{multline*}
with $\widehat{\rm E}(Y_i^{(k)}\mid\bm Z^{(k)})$
denoting the $i$th entry of $\widehat{\rm E}(\bm Y^{(k)}\mid\bm Z^{(k)})$.


We thus approximate 
${\rm E}(\hat{\bm c}^{(k)}\mid\bm Z^{(k)})$ and
${\rm var}(\hat{\bm c}^{(k)}\mid\bm Z^{(k)})$, respectively, by
\begin{equation}\label{eq:hat_E_hat_c}
    \widehat{\rm E}(\hat{\bm c}^{(k)}\mid\bm Z^{(k)})=
    J_k^{-1}
    (\bm\Omega^{(k)}+\lambda^{(k)}\bm W)^{-1}
    \bm\Psi\bm P^{(k)}\bm Z^{(k)}
    \bm Z^{(k)\top}\bm D\bm\Phi^{(k)}\hat{\bm c}^{(k)}
\end{equation}
and
\begin{multline}\label{eq:hat_var_hat_c}
    \widehat{\rm var}(\hat{\bm c}^{(k)}\mid\bm Z^{(k)})=
    \\
    (\bm\Omega^{(k)}+\lambda^{(k)}\bm W)^{-1}
    \bm\Psi\bm P^{(k)}\bm Z^{(k)}
    \widehat{\rm var}(\bm Y^{(k)}\mid\bm Z^{(k)})
    \bm Z^{(k)\top}\bm P^{(k)\top}\bm\Psi
    (\bm\Omega^{(k)}+\lambda^{(k)}\bm W)^{-1}.
\end{multline}

\subsection{Inverse matrices in Eqs. \eqref{eq:l_c_delta} and \eqref{eq:pl_c_delta}}
\label{sec:inverse_matrix}

According to the inverse formula of $2\times2$ partitioned matrix
\citep[see, e.g.,][Theorem 7.1]{Schott2017},
\begin{align*}
    {\rm cov}^{-1}(\hat{\bm c}^{(0)}, \hat{\bm\delta}\mid\mathcal Z)
    =
    \begin{bmatrix}
        {\rm var}(\hat{\bm c}^{(0)}\mid\bm Z^{(0)}) &
        \mathbf 1_K^\top\otimes{\rm var}(\hat{\bm c}^{(0)}\mid\bm Z^{(0)})\\
        \mathbf 1_K\otimes{\rm var}(\hat{\bm c}^{(0)}\mid\bm Z^{(0)}) &
        {\rm var}(\hat{\bm\delta}\mid\mathcal Z)
    \end{bmatrix}^{-1}
    =
    \begin{bmatrix}
        \bm B_{11} &\bm B_{12}\\
        \bm B_{12}^\top &\bm B_{22}
    \end{bmatrix}
\end{align*}
where 
\begin{align*}
    \bm B_{11}
    &=
    [{\rm var}(\hat{\bm c}^{(0)}\mid\bm Z^{(0)})-
    \{\mathbf 1_K^\top\otimes{\rm var}(\hat{\bm c}^{(0)}\mid\bm Z^{(0)})\}
    {\rm var}^{-1}(\hat{\bm\delta}\mid\mathcal Z)
    \{\mathbf 1_K\otimes{\rm var}(\hat{\bm c}^{(0)}\mid\bm Z^{(0)})\}]^{-1}
    \\
    &=\sum_{k=0}^K{\rm var}^{-1}(\hat{\bm c}^{(k)}\mid\bm Z^{(k)}),
    \quad\text{(by \eqref{eq:inv_var_delta_hat})}
\end{align*}
\begin{align*}
    \bm B_{22}
    &=
    [{\rm var}(\hat{\bm\delta}\mid\mathcal Z)-
    \{\mathbf 1_K\otimes{\rm var}(\hat{\bm c}^{(0)}\mid\bm Z^{(0)})\}
    {\rm var}^{-1}(\hat{\bm c}^{(0)}\mid\bm Z^{(0)})
    \{\mathbf 1_K^\top\otimes{\rm var}(\hat{\bm c}^{(0)}\mid\bm Z^{(0)})\}]^{-1}
    \\
    &=
    {\rm diag}\big\{
        {\rm var}^{-1}(\hat{\bm c}^{(1)}\mid\bm Z^{(1)}),\ldots,
        {\rm var}^{-1}(\hat{\bm c}^{(K)}\mid\bm Z^{(K)})
    \big\},
\end{align*}
and
\begin{align*}
    \bm B_{12}
    &=
    -\bm B_{11}
    \{\mathbf 1_K^\top\otimes{\rm var}(\hat{\bm c}^{(0)}\mid\bm Z^{(0)})\}
    {\rm var}^{-1}(\hat{\bm\delta}\mid\mathcal Z)
    \\
    &=-[{\rm var}^{-1}(\hat{\bm c}^{(1)}\mid\bm Z^{(1)}),\ldots,
    {\rm var}^{-1}(\hat{\bm c}^{(K)}\mid\bm Z^{(K)})].
    \quad\text{(by \eqref{eq:inv_var_delta_hat})}
\end{align*}
Its empirical counterpart follows:
\begin{align*}
    \widehat{\rm cov}^{-1}(\hat{\bm c}^{(0)}, \hat{\bm\delta}\mid\mathcal Z)
    =
    \begin{bmatrix}
        \widehat{\bm B}_{11} &\widehat{\bm B}_{12}\\
        \widehat{\bm B}_{12}^\top &\widehat{\bm B}_{22}
    \end{bmatrix},
\end{align*}
where 
$
    \widehat{\bm B}_{11}
    =\sum_{k=0}^K\widehat{\rm var}^{-1}(\hat{\bm c}^{(k)}\mid\bm Z^{(k)})
$,
$
    \widehat{\bm B}_{22}
    =
    {\rm diag}\big\{
        \widehat{\rm var}^{-1}(\hat{\bm c}^{(1)}\mid\bm Z^{(1)}),\ldots,
        \widehat{\rm var}^{-1}(\hat{\bm c}^{(K)}\mid\bm Z^{(K)})
    \big\}
$,
and
$
    \widehat{\bm B}_{12}
    =
    -[\widehat{\rm var}^{-1}(\hat{\bm c}^{(1)}\mid\bm Z^{(1)}),\ldots,
    \widehat{\rm var}^{-1}(\hat{\bm c}^{(K)}\mid\bm Z^{(K)})]
$,
with $\widehat{\rm var}(\hat{\bm c}^{(k)}\mid\bm Z^{(k)})$ defined as \eqref{eq:hat_var_hat_c}.

\subsection{Constants in Algorithm \ref{algo:AOTL}}\label{sec:constants}

As suggested by \citet[Remark 4]{GaiffasLecue2011},
$b_1=4(1+9b_3)$,
with $b_3$ as an upper bound of 
$\max\{\sup_i|Y_i^{(0)}|,\max_{\beta\in\mathbb B}\sup_i|\langle X_i^{(0)},\beta\rangle_{L^2}|\}$ 
and $\mathbb B$ as defined in Algorithm \ref{algo:AOTL}.
We understand that 
$\max\{\sup_i|Y_i^{(0)}|,\max_{\beta\in\mathbb B}\sup_i|\langle X_i^{(0)},\beta\rangle_{L^2}|\}$ 
is unknown and might be unbounded in practice.
A surrogate is 
$\max\{\max_{1\leq i\leq n_0}|Y_i^{(0)}|,\max_{\beta\in\mathbb B}\max_{1\leq i\leq n_0}|\langle \widehat X_i^{(0)},\beta\rangle_{L^2}|\}$.
As for $b_2$,
it is set to be $b_3[\{\ln(K+1)+\alpha\}/n_0]^{1/2}$
for confidence level $\alpha\in(0,1)$.

\subsection{Lemmas}

\begin{lemma}[\citealp{MarshallOlkinArnold2011}, pp. 340--341]
\label{lemma:trace_inequality}
    Let $\gamma_i(\cdot)$
    denote the $i$th largest singular value of a matrix.
    If $\bm A$ and $\bm B$ are both $n\times n$ matrices,
    then
    $$
        \sum_{i=1}^n\gamma_i(\bm A\bm B)\leq\sum_{i=1}^n\gamma_i(\bm A)\gamma_i(\bm B).
    $$
    Further, if $\bm A$ and $\bm B$ are both Hermitian,
    then
    $$
        \sum_{i=1}^n\gamma_i(\bm A\bm B)\geq\sum_{i=1}^n\gamma_i(\bm A)\gamma_{n-i+1}(\bm B).
    $$
\end{lemma}

\begin{lemma}[\citealp{MerikoskiKumar2004}, Theorems 1 and 7]
\label{lemma:eigen_inequality}
    Let $\gamma_i(\cdot)$
    denote the $i$th largest singular value of a matrix.
    Suppose $n\times n$ matrices $\bm A$ and $\bm B$ are both Hermitian.
    If $1\leq k\leq i\leq n$ and $1\leq\ell\leq n-i+1$,
    then
    $$
        \gamma_{i+\ell-1}(\bm A)+\gamma_{n-\ell+1}(\bm B)
        \leq
        \gamma_i(\bm A+\bm B)
        \leq
        \gamma_{i-k+1}(\bm A)+\gamma_k(\bm B).
    $$
    If $\bm A$ and $\bm B$ are further assumed to be positive semidefinite,
    then
    $$
        \gamma_{i+\ell-1}(\bm A)\gamma_{n-\ell+1}(\bm B)
        \leq
        \gamma_i(\bm A\bm B)
        \leq
        \gamma_{i-k+1}(\bm A)\gamma_k(\bm B).
    $$
\end{lemma}

\begin{lemma}\label{lemma:aUBUa}
    For ${\bm a}_k\in\mathbb R^M$, $k=0,\ldots,K$,
    and positive definite $\bm B\in\mathbb R^{M\times M}$,
    \begin{align*}
        \bigg\|&
        \bm B^{\frac{1}{2}}
        \bigg\{
        \sum_{k=0}^K
        {\rm var}^{-1}(\hat{\bm c}^{(k)}\mid\bm Z^{(k)})
        \bigg\}^{-1}
        \sum_{k=1}^K
        {\rm var}^{-1}(\hat{\bm c}^{(k)}\mid\bm Z^{(k)})\bm a_k
        \bigg\|_2^2
        \leq
        K\sum_{k=1}^K\bm a_k^\top\bm B\bm a_k
    \end{align*}
    and
    \begin{align*}
        \bigg\|&
        \bm B^{\frac{1}{2}}
        \bigg\{
        \sum_{k=0}^K
        \widehat{\rm var}^{-1}(\hat{\bm c}^{(k)}\mid\bm Z^{(k)})
        \bigg\}^{-1}
        \sum_{k=1}^K
        \widehat{\rm var}^{-1}(\hat{\bm c}^{(k)}\mid\bm Z^{(k)})\bm a_k
        \bigg\|_2^2
        \leq
        K\sum_{k=1}^K\bm a_k^\top\bm B\bm a_k.
    \end{align*}
\end{lemma}

\begin{lemma}[\citealp{LiHsing2007}, Theorem 4 and its proof]\label{lemma:bound_X}
    With \eqref{cond:trajectories} and \eqref{cond:rho_J} and fixed $k$,
    $$
    \|\widehat C^{(k)}-C^{(k)}\|_{\rm HS}^2
    =
    \int_0^1\int_0^1\{\widehat C^{(k)}(s,t)-C^{(k)}(s,t)\}^2{\rm d}s{\rm d}t
    =
    O(n^{-1}+\rho+ J^{-2}\rho^{-1/4}),
    $$
    where 
    $\|\cdot\|_{\rm HS}$ is the Hilbert-Schmidt norm.
    Moreover,
    $$
    {\rm E}\left(
    \|\widehat X_1^{(k)}\|_{L^2}^4
    \right)
    =O(1)
    \quad\text{and}\quad
    {\rm E}\left(
    \|\widehat X_1^{(k)}-X_1^{(k)}\|_{L^2}^4
    \right)
    =O(\rho^2+ J^{-2}\rho^{-1/2}).
    $$
\end{lemma}


\begin{lemma}\label{lemma:pred_err_internal}
    For each $k$,
    under conditions 
    \eqref{cond:trajectories},
    \eqref{cond:rho_J},
    and \eqref{cond:lambda_n},
    $$
    {\rm E}\left(
    \bigg\{\frac{1}{n}\sum_{i=1}^n
    \langle\beta^{(k)}-\hat\beta^{(k)},\widehat X_i^{(k)}\rangle_{L^2}^2
    \bigg\}^2
    \mid\bm Z^{(k)}
    \right)
    =O_p(\lambda^2+\rho^2+n^{-2}\lambda^{-1/2}+J^{-2}\rho^{-1/2}).
    $$
\end{lemma}

\begin{lemma}\label{lemma:pred_err_external}
    Focus on the target data.
    Suppose there is a testing set (indexed by $(00)$)
    consisting of $n$ new independent curves,
    say $X_1^{(00)},\ldots,X_n^{(00)}$,
    from the target population 
    but independent of $X_1^{(0)},\ldots,X_n^{(0)}$.
    Under conditions 
    \eqref{cond:trajectories},
    \eqref{cond:rho_J},
    and \eqref{cond:lambda_n},
    $$
    {\rm E}\left(
    \frac{1}{n}\sum_{i=1}^n
    \langle\beta^{(0)}-\hat\beta^{(0)},\widehat X_i^{(00)}\rangle_{L^2}^2
    \mid\bm Z^{(0)}
    \right)
    =O_p(\lambda+\rho+n^{-1}\lambda^{-1/4}+J^{-1}\rho^{-1/4}).
    $$
\end{lemma}

\begin{lemma}\label{lemma:bounds_c_Ec}
    Conditioning on 
    \eqref{cond:trajectories}
    and \eqref{cond:xi},
    $$
        {\rm E}
        (\|\bm\phi^\top
        \{{\rm E}(\hat{\bm c}^{(k)}\mid\bm Z^{(k)})
        -\widehat{\rm E}(\hat{\bm c}^{(k)}\mid\bm Z^{(k)})\}
        \|_{\widehat C^{(0)}}^2
        \mid\mathcal Z)
        =O_p(n^{-1}\lambda^{-1/4}J^\xi)
    $$
    and
    $$
        {\rm E}
        (\|\bm\phi^\top
        \{\hat{\bm c}^{(k)}
        -\widehat{\rm E}(\hat{\bm c}^{(k)}\mid\bm Z^{(k)})\}
        \|_{\widehat C^{(0)}}^2
        \mid\mathcal Z)
        =O_p(n^{-1}\lambda^{-1/4}J^\xi).
    $$
\end{lemma}

\subsection{Proofs}

\begin{proof}[Proof of Lemma \ref{lemma:aUBUa}]
    \begin{align*}
        &\bigg\|
        \bm B^{\frac{1}{2}}
        \bigg\{
        \sum_{k=0}^K
        {\rm var}^{-1}(\hat{\bm c}^{(k)}\mid\bm Z^{(k)})
        \bigg\}^{-1}
        \sum_{k=1}^K
        {\rm var}^{-1}(\hat{\bm c}^{(k)}\mid\bm Z^{(k)})\bm a_k
        \bigg\|_2^2
        \\
        &\leq
        K\sum_{k=1}^K
        \bigg\|
        \bm B^{\frac{1}{2}}
        \bigg\{
        \sum_{k=0}^K
        {\rm var}^{-1}(\hat{\bm c}^{(k)}\mid\bm Z^{(k)})
        \bigg\}^{-1}
        {\rm var}^{-1}(\hat{\bm c}^{(k)}\mid\bm Z^{(k)})\bm a_k
        \bigg\|_2^2
        \quad\text{(Cauchy-Schwarz)}
        \\
        &\leq
        K
        \sum_{k=1}^K
        \bigg\{
        \gamma_1
        \bigg(
        \bigg\{
        \sum_{k=0}^K
        {\rm var}^{-1}(\hat{\bm c}^{(k)}\mid\bm Z^{(k)})
        \bigg\}^{-1}
        {\rm var}^{-1}(\hat{\bm c}^{(k)}\mid\bm Z^{(k)})
        \bigg)
        \bigg\}^2
        \bm a_k^\top\bm B\bm a_k
        \quad\text{(Lemma \ref{lemma:trace_inequality})}
        \\
        &\leq
        K\sum_{k=1}^K\bm a_k^\top\bm B\bm a_k.
    \end{align*}
    This proof applies too even when
    ${\rm var}(\cdot)$ is replaced with its empirical counterpart.
\end{proof}

\begin{proof}[Proof of Lemma \ref{lemma:pred_err_internal}]
    Let $\gamma_i(\bm A)$
    denote the $i$th largest singular value of real matrix $\bm A$.
    Lemma \ref{lemma:bound_X} implies that,
    for each fixed $i$,
    with diverging $J$ and zero-convergent $\rho$,
    $\gamma_i(\bm\Omega^{(k)})$ approaches the $i$th eigenvalue of $C^{(k)}(s,t)$,
    which is a positive constant.
    That is, 
    $\gamma_i(\bm\Omega^{(k)})=O_p(1)$.
    Meanwhile,
    $$
        \gamma_i((\bm\Omega^{(k)}+\lambda\bm W)^{-1})
        =
        (\gamma_{M-i+1}(\bm\Omega^{(k)}+\lambda\bm W))^{-1}
        \leq
        (\lambda(\lfloor i/2\rfloor2\pi)^4)^{-1}
        \leq
        (\lambda((i-1)\pi)^4)^{-1}.
    $$
    Apparently,
    $0<\gamma_i(\bm\Omega^{(k)}(\bm\Omega^{(k)}+\lambda\bm W)^{-1})\leq 1$ and,
    by Lemma \ref{lemma:eigen_inequality},
    $$
        \gamma_i(\bm\Omega^{(k)}(\bm\Omega^{(k)}+\lambda\bm W)^{-1})
        \leq
        \gamma_1(\bm\Omega^{(k)})
        \gamma_i((\bm\Omega^{(k)}+\lambda\bm W)^{-1})
        \leq
        \gamma_1(\bm\Omega^{(k)})
        (\lambda\pi^4(i-1)^4)^{-1}.
    $$
    Combining above identities,
    \begin{align*}
        {\rm tr}&
        \{
            \bm\Omega^{(k)}
            (\bm\Omega^{(k)}+\lambda\bm W)^{-1}
        \}
        \\
        &=
        \sum_{i=1}^M\gamma_i(\bm\Omega^{(k)}(\bm\Omega^{(k)}+\lambda\bm W)^{-1})
        \\
        &\leq
        \sum_{i=1}^{\min\{M,\lfloor\pi^{-1}(\gamma_1(\bm\Omega^{(k)}))^{1/4}\lambda^{-1/4}\rfloor\}}
        1
        +
        \sum_{i=\min\{M,\lfloor\pi^{-1}(\gamma_1(\bm\Omega^{(k)}))^{1/4}\lambda^{-1/4}\rfloor\}+1}^M
        \frac{\gamma_1(\bm\Omega^{(k)})(i-1)^{-4}}{\pi^4\lambda}
        \\
        &\leq 
        \min\{\pi^{-1}(\gamma_1(\bm\Omega^{(k)}))^{1/4}\lambda^{-1/4}, M\}
        +
        c_1\gamma_1(\bm\Omega^{(k)})
        \\
        &=
        O_p(\min\{\lambda^{-1/4}, J\})
        =O_p(\lambda^{-1/4}),
    \end{align*}
    with a constant $c_1>0$.
    Further, for each positive integer $m$
    and sufficiently small $\lambda>0$,
    \begin{equation}\label{eq:tr_omega_omega+lambdaW_inv}
        {\rm tr}
        [
        \{
            \bm\Omega^{(k)}
            (\bm\Omega^{(k)}+\lambda\bm W)^{-1}
        \}^m
        ]
        \leq
        {\rm tr}
        \{
            \bm\Omega^{(k)}
            (\bm\Omega^{(k)}+\lambda\bm W)^{-1}
        \}
        =O_p(\lambda^{-1/4}).
    \end{equation}

    Secondly, recall that $\check\beta^{(k)}=\bm\phi^\top\check{\bm c}^{(k)}$ 
    is the projection of $\beta^{(k)}$ onto the linear space spanned by $\phi_1,\ldots,\phi_M$.
    Decomposing $Y_i^{(k)}$,
    we have 
    \begin{align*}
        Y_i^{(k)} 
        = \langle\widehat X_i^{(k)},\check\beta^{(k)}\rangle_{L^2}+r_i^{(k)}+\varepsilon_i^{(k)},
    \end{align*}
    where $r_i^{(k)}=
    \langle X_i^{(k)}-\widehat X_i^{(k)},\beta^{(k)}\rangle_{L^2}$.
    Writing 
    $\bm r^{(k)}=[r_1^{(k)},\ldots,r_n^{(k)}]^\top\in\mathbb R^n$
    and 
    $\bm\varepsilon^{(k)}=[\varepsilon_1^{(k)},\ldots,\varepsilon_n^{(k)}]^\top\in\mathbb R^n$,
    \begin{align*}
        \hat\beta^{(k)}
        &=\bm\phi^\top\hat{\bm c}^{(k)}
        \\
        &=n^{-1}\bm\phi^\top(\bm\Omega^{(k)}+\lambda\bm W)^{-1}
        \bm P^{(k)}\bm Z^{(k)}
        (\bm Z^{(k)\top}\bm P^{(k)\top}\check{\bm c}^{(k)}+\bm r^{(k)}+\bm\varepsilon^{(k)})
        \\
        &=\bm\phi^\top({\bm c}_\lambda^{(k)}+{\bm c}_r^{(k)}+{\bm c}_\varepsilon^{(k)}),
    \end{align*}
    where 
    \begin{align*}
        {\bm c}_\lambda^{(k)}
        &=
        (\bm\Omega^{(k)}+\lambda\bm W)^{-1}
        \bm\Omega^{(k)}
        \check{\bm c}^{(k)},
        \\
        \bm c_r^{(k)}
        &=
        n^{-1}(\bm\Omega^{(k)}+\lambda\bm W)^{-1}
        \bm P^{(k)}\bm Z^{(k)}
        \bm r^{(k)},
        \\
        \bm c_\varepsilon^{(k)}
        &=
        n^{-1}(\bm\Omega^{(k)}+\lambda\bm W)^{-1}
        \bm P^{(k)}\bm Z^{(k)}
        \bm\varepsilon^{(k)},
    \end{align*}
    and $\hat{\bm c}^{(k)}={\bm c}_\lambda^{(k)}+{\bm c}_r^{(k)}+{\bm c}_\varepsilon^{(k)}$.
    A further decomposition follows:
    \begin{align*}
        \bigg\{&\frac{1}{n}\sum_{i=1}^n
        \langle\beta^{(k)}-\hat\beta^{(k)},\widehat X_i^{(k)}\rangle_{L^2}^2
        \bigg\}^2
        \\
        &=
        \bigg\{\frac{1}{n}\sum_{i=1}^n
        \langle\beta^{(k)}-\bm\phi^\top({\bm c}_\lambda^{(k)}+{\bm c}_r^{(k)}+{\bm c}_\varepsilon^{(k)}),
        \widehat X_i^{(k)}\rangle_{L^2}^2
        \bigg\}^2
        \\
        &\leq
        \bigg\{
        \frac{3}{n}\sum_{i=1}^n
        \langle\beta^{(k)}-\bm\phi^\top{\bm c}_\lambda^{(k)},
        \widehat X_i^{(k)}\rangle_{L^2}^2
        +
        \frac{3}{n}\sum_{i=1}^n
        \langle\bm\phi^\top{\bm c}_r^{(k)},
        \widehat X_i^{(k)}\rangle_{L^2}^2
        +
        \frac{3}{n}\sum_{i=1}^n
        \langle\bm\phi^\top{\bm c}_\varepsilon^{(k)},
        \widehat X_i^{(k)}\rangle_{L^2}^2
        \bigg\}^2
        \\
        &\leq
        3\bigg\{
        \frac{3}{n}\sum_{i=1}^n
        \langle\beta^{(k)}-\bm\phi^\top{\bm c}_\lambda^{(k)},
        \widehat X_i^{(k)}\rangle_{L^2}^2
        \bigg\}^2
        +
        3\bigg\{
        \frac{3}{n}\sum_{i=1}^n
        \langle\bm\phi^\top{\bm c}_r^{(k)},
        \widehat X_i^{(k)}\rangle_{L^2}^2
        \bigg\}^2
        +
        \\
        &\qquad
        3\bigg\{
        \frac{3}{n}\sum_{i=1}^n
        \langle\bm\phi^\top{\bm c}_\varepsilon^{(k)},
        \widehat X_i^{(k)}\rangle_{L^2}^2
        \bigg\}^2,
    \end{align*}
    where, conditioning on $\bm Z^{(k)}$, 
    the expectation of the third term is proportional to
    \begin{align}
        \notag
        {\rm E}&
        \Bigg(
        \bigg\{
        \frac{1}{n}\sum_{i=1}^n
        \langle\bm\phi^\top{\bm c}_\varepsilon^{(k)},
        \widehat X_i^{(k)}\rangle_{L^2}^2
        \bigg\}^2
        \mid\bm Z^{(k)}
        \Bigg)
        \\ \notag
        &=
        {\rm E}
        (
        (
        n^{-1}
        \bm\varepsilon^{(k)\top}
        \bm A_1^{(k)}
        \bm\varepsilon^{(k)}
        )^2
        \mid\bm Z^{(k)}
        )
        \\ \notag
        &=
        {\rm E}
        (
        \{
        {\rm tr}
        (
        n^{-1}
        \bm A_1^{(k)}
        \bm\varepsilon^{(k)}
        \bm\varepsilon^{(k)\top}
        )
        \}^2
        \mid\bm Z^{(k)}
        )
        \\ \notag
        &\leq
        {\rm E}
        \Bigg(
        \bigg\{
        \frac{1}{n}
        \sum_{i=1}^n a_{ii}(\bm A_1^{(k)})
        \varepsilon_i^{(k)2}
        \bigg\}^2
        \mid\bm Z^{(k)}
        \Bigg)
        \\ \notag
        &=
        \frac{\{{\rm tr}(\bm\Omega^{(k)2}(\bm\Omega^{(k)}+\lambda\bm W)^{-2})\}^2}{n^2}
        {\rm E}
        \Bigg(
        \bigg\{
        \sum_{m=1}^M
        \frac{a_{ii}(\bm A_1^{(k)})}
            {{\rm tr}(\bm\Omega^{(k)2}(\bm\Omega^{(k)}+\lambda\bm W)^{-2})}
        \varepsilon_{i_m}^{(k)2}
        \bigg\}^2
        \mid\bm Z^{(k)}
        \Bigg)
        \\ \label{eq:bound_c_eps}
        &=
        O_p(n^{-2}\lambda^{-1/2}),
        \quad\text{(refer to \eqref{eq:tr_omega_omega+lambdaW_inv} and 
        \citet[Theorem 1]{Pruitt1966})}
    \end{align}
    with $n\times n$ matrix
    $$
        \bm A_1^{(k)}=
        (
        n^{-1/2}
        \bm Z^{(k)\top}\bm P^{(k)\top}
        )
        (\bm\Omega^{(k)}+\lambda\bm W)^{-1}\bm\Omega^{(k)}
        (\bm\Omega^{(k)}+\lambda\bm W)^{-1}
        (
        n^{-1/2}
        \bm P^{(k)}\bm Z^{(k)}
        )
    $$
    and $a_{ii}(\bm A_1^{(k)})$ ($\in(0,1)$) denoting 
    the $(i,i)$-entry of $\bm A_1^{(k)}$.
    The proof is completed by bounding the remaining two terms.
    In particular,
    the second term is
    \begin{align}
        \notag
        {\rm E}&
        \Bigg(
        \bigg\{
        \frac{1}{n}\sum_{i=1}^n
        \langle\bm\phi^\top{\bm c}_r^{(k)},
        \widehat X_i^{(k)}\rangle_{L^2}^2
        \bigg\}^2
        \mid\bm Z^{(k)}
        \Bigg)
        \\ \notag
        &=
        {\rm E}
        \Bigg(
        \bigg\{
        n^{-2}
        \bm r^{(k)\top}
        \bm Z^{(k)\top}\bm P^{(k)\top}
        (\bm\Omega^{(k)}+\lambda\bm W)^{-1}\bm\Omega^{(k)}
        (\bm\Omega^{(k)}+\lambda\bm W)^{-1}\bm P^{(k)}\bm Z^{(k)}
        \bm r^{(k)}
        \bigg\}^2
        \mid\bm Z^{(k)}
        \Bigg)
        \\ \notag
        &\leq
        {\rm E}
        \Big(
        \big\{
        n^{-1}
        \|(\bm\Omega^{(k)}+\lambda\bm W)^{-1}\bm\Omega^{(k)}\|_2^2
        \|\bm r^{(k)}\|_2^2
        \big\}^2
        \mid\bm Z^{(k)}
        \Big)
        \\ \notag
        &\leq
        n^{-2}
        {\rm E}
        \Big(
        \big(
        \|\bm r^{(k)}\|_2^2
        \big)^2
        \mid\bm Z^{(k)}
        \Big)
        \\ \notag
        &\leq
        n^{-2}
        {\rm E}
        \Bigg(
        \bigg(
        \sum_{i=1}^n\|\widehat X_i^{(k)}-X_i^{(k)}\|_{L^2}^2\|\beta^{(k)}\|_{L^2}^2
        \bigg)^2
        \mid\bm Z^{(k)}
        \Bigg)
        \\ \notag
        &\leq
        \|\beta^{(k)}\|_{L^2}^4
        \frac{1}{n}\sum_{i=1}^n
        {\rm E}
        \left(
        \|\widehat X_i^{(k)}-X_i^{(k)}\|_{L^2}^4
        \mid\bm Z^{(k)}
        \right)
        \\ \label{eq:bound_c_r}
        &=
        O(\rho^2+J^{-2}\rho^{-1/2})
        \quad\text{(Lemma \ref{lemma:bound_X})}
    \end{align}
    and the first term is
    \begin{align}
        \notag
        \bigg(&
        \frac{1}{n}\sum_{i=1}^n
        \langle\beta^{(k)}-\bm\phi^\top{\bm c}_\lambda^{(k)},
        \widehat X_i^{(k)}\rangle_{L^2}^2
        \bigg)^2
        \\ \notag
        &\leq
        \bigg(
        \frac{1}{n}\sum_{i=1}^n
        \langle\beta^{(k)}-\bm\phi^\top{\bm c}_\lambda^{(k)},
        \widehat X_i^{(k)}\rangle_{L^2}^2
        +
        \lambda
        {\bm c}_\lambda^{(k)\top}
        \bm\phi
        \bm W
        \bm\phi^\top
        {\bm c}_\lambda^{(k)}
        \bigg)^2
        \\ \notag
        &\leq
        \bigg(
        \frac{1}{n}\sum_{i=1}^n
        \langle\beta^{(k)}-\bm\phi^\top\check{\bm c}^{(k)},
        \widehat X_i^{(k)}\rangle_{L^2}^2
        +
        \lambda
        \check{\bm c}^{(k)\top}
        \bm\phi
        \bm W
        \bm\phi^\top
        \check{\bm c}^{(k)}
        \bigg)^2
        \\ \label{eq:bound_c_lambda}
        &=
        (
        \lambda
        \check{\bm c}^{(k)\top}
        \bm\phi
        \bm W
        \bm\phi^\top
        \check{\bm c}^{(k)}
        )^2
        \leq
        \lambda^2\|\beta^{(k)''}\|_{L^2}^2
        =O(\lambda^2),
    \end{align}
    because ${\bm c}_\lambda^{(k)}$ is the solution to
    $$
    \min_{\bm c\in\mathbb R^M} 
    n^{-1}\sum_{i=1}^n
    \langle\beta^{(k)}-\bm\phi^\top{\bm c},
    \widehat X_i^{(k)}\rangle_{L^2}^2
    +
    \lambda
    {\bm c}^\top
    \bm\phi
    \bm W
    \bm\phi^\top
    {\bm c}.
    $$
\end{proof}

\begin{proof}[Proof of Lemma \ref{lemma:pred_err_external}]
    Lemma \ref{lemma:bound_X} implies that
    $$
        \|\bm\Omega^{(0)}-\bm\Omega^{(00)}\|_{\rm HS}^2
        =O_p(n^{-1}+\rho+ J^{-2}\rho^{-1/4}).
    $$
    Namely,
    for arbitrary nonzero $\bm a\in\mathbb R^M$,
    $\bm a^\top\bm\Omega^{(0)}\bm a$ and 
    $\bm a^\top\bm\Omega^{(00)}\bm a$ become indistinguishable in the limit.
    Lemma \ref{lemma:pred_err_internal}
    helps finalize the proof,
    after noting that
    \begin{align*}
        \frac{1}{n}&\sum_{i=1}^n
        \langle\beta^{(0)}-\hat\beta^{(0)},\widehat X_i^{(00)}\rangle_{L^2}^2
        \\
        &=
        (\check{\bm c}^{(0)}-\hat{\bm c}^{(0)})^\top
        \bm\Omega^{(0)\frac{1}{2}}
        (\bm\Omega^{(0)-\frac{1}{2}}\bm\Omega^{(00)}\bm\Omega^{(0)-\frac{1}{2}})
        \bm\Omega^{(0)\frac{1}{2}}
        (\check{\bm c}^{(0)}-\hat{\bm c}^{(0)})
        \\
        &\leq
        \gamma_1(\bm\Omega^{(0)-\frac{1}{2}}\bm\Omega^{(00)}\bm\Omega^{(0)-\frac{1}{2}})
        (\check{\bm c}^{(0)}-\hat{\bm c}^{(0)})^\top
        \bm\Omega^{(0)}
        (\check{\bm c}^{(0)}-\hat{\bm c}^{(0)})
        \quad\text{(Lemma \ref{lemma:trace_inequality})}
        \\
        &=
        \bigg(\max_{\bm a\neq\bm 0}
        \frac{\bm a^\top\bm\Omega^{(00)}\bm a}{\bm a^\top\bm\Omega^{(0)}\bm a}
        \bigg)
        \frac{1}{n}\sum_{i=1}^n
        \langle\beta^{(0)}-\hat\beta^{(0)},\widehat X_i^{(0)}\rangle_{L^2}^2.
    \end{align*}
\end{proof}

\begin{proof}[Proof of Lemma \ref{lemma:bounds_c_Ec}]
    Let $\bm A_2^{(k)}$ denote square matrix
    $
    (n^{-\frac{1}{2}}\bm Z^{(k)\top}\bm P^{(k)\top})
    (\bm\Omega^{(k)}+\lambda\bm W)^{-1}
    \bm\Omega^{(0)}
    (\bm\Omega^{(k)}+\lambda\bm W)^{-1}
    (n^{-\frac{1}{2}}\bm P^{(k)}\bm Z^{(k)}).
    $
    Observe that
    \begin{align}
        \notag
        {\rm tr}(\bm A_2^{(k)})
        &=
        {\rm tr}\{
        (\bm\Omega^{(k)}+\lambda\bm W)^{-1}
        \bm\Omega^{(k)}
        (\bm\Omega^{(k)}+\lambda\bm W)^{-1}
        \bm\Omega^{(0)}
        \}
        \\\notag
        &=
        {\rm tr}\{
        (\bm\Omega^{(k)}+\lambda\bm W)^{-1}
        \bm\Omega^{(k)}
        (\bm\Omega^{(k)}+\lambda\bm W)^{-1}
        \bm\Omega^{(k)\frac{1}{2}}
        (\bm\Omega^{(k)-\frac{1}{2}}
        \bm\Omega^{(0)}
        \bm\Omega^{(k)-\frac{1}{2}})
        \bm\Omega^{(k)\frac{1}{2}}
        \}
        \\\notag
        &\leq
        \gamma_1(\bm\Omega^{(k)-\frac{1}{2}}
        \bm\Omega^{(0)}
        \bm\Omega^{(k)-\frac{1}{2}})
        {\rm tr}\{
        (\bm\Omega^{(k)}+\lambda\bm W)^{-1}
        \bm\Omega^{(k)}
        (\bm\Omega^{(k)}+\lambda\bm W)^{-1}
        \bm\Omega^{(k)}
        \}
        \quad\text{(Lemma \ref{lemma:trace_inequality})}
        \\\label{eq:tr_A2}
        &=O_p(\lambda^{-1/4}J^\xi)
        \quad\text{(refer to \eqref{eq:tr_omega_omega+lambdaW_inv} and \eqref{cond:xi})}
    \end{align}
    and
    \begin{align}
        \notag
        \gamma_1(\bm A_2^{(k)})
        &=
        \gamma_1\{
        (\bm\Omega^{(k)}+\lambda\bm W)^{-1}
        \bm\Omega^{(k)}
        (\bm\Omega^{(k)}+\lambda\bm W)^{-1}
        \bm\Omega^{(0)}
        \}
        \\\notag
        &=
        \gamma_1\{
        (\bm\Omega^{(k)}+\lambda\bm W)^{-1}
        \bm\Omega^{(k)}
        (\bm\Omega^{(k)}+\lambda\bm W)^{-1}
        \bm\Omega^{(k)\frac{1}{2}}
        (\bm\Omega^{(k)-\frac{1}{2}}
        \bm\Omega^{(0)}
        \bm\Omega^{(k)-\frac{1}{2}})
        \bm\Omega^{(k)\frac{1}{2}}
        \}
        \\\notag
        &\leq
        \gamma_1(\bm\Omega^{(k)-\frac{1}{2}}
        \bm\Omega^{(0)}
        \bm\Omega^{(k)-\frac{1}{2}})
        \gamma_1\{
        (\bm\Omega^{(k)}+\lambda\bm W)^{-1}
        \bm\Omega^{(k)}
        (\bm\Omega^{(k)}+\lambda\bm W)^{-1}
        \bm\Omega^{(k)}
        \}
        \quad\text{(Lemma \ref{lemma:eigen_inequality})}
        \\\label{eq:eigen_A2}
        &=O_p(J^\xi).
        \quad\text{(refer to \eqref{eq:tr_omega_omega+lambdaW_inv} and \eqref{cond:xi})}
    \end{align}

    Recalling 
    $\hat{\bm c}^{(k)}$ in \eqref{eq:hat_beta_k} and
    $\widehat{\rm E}(\hat{\bm c}^{(k)}\mid\bm Z^{(k)})$ in \eqref{eq:hat_E_hat_c},
    \begin{align}
        \notag
        {\rm E}&
        (\|\bm\phi^\top
        \{{\rm E}(\hat{\bm c}^{(k)}\mid\bm Z^{(k)})
        -\widehat{\rm E}(\hat{\bm c}^{(k)}\mid\bm Z^{(k)})\}
        \|_{\widehat C^{(0)}}^2
        \mid\mathcal Z)
        \\\notag
        &=
        {\rm E}(\|
        n^{-\frac{1}{2}}
        (n^{-\frac{1}{2}}\bm Z^{(0)\top}\bm P^{(0)\top})
        (\bm\Omega^{(k)}+\lambda\bm W)^{-1}
        \\\notag
        &\qquad
        (n^{-\frac{1}{2}}\bm P^{(k)}\bm Z^{(k)})
        \{{\rm E}(\bm Y^{(k)}\mid\bm Z^{(k)})-\bm Z^{(k)\top}\bm P^{(k)\top}\hat{\bm c}^{(k)}\}
        \|_2^2\mid\mathcal Z)
        \\\notag
        &=
        n^{-1}
        {\rm E}(
        \{{\rm E}(\bm Y^{(k)\top}\mid\bm Z^{(k)})
        -\hat{\bm c}^{(k)\top}\bm P^{(k)}\bm Z^{(k)}\}
        \bm A_2^{(k)}
        \{{\rm E}(\bm Y^{(k)}\mid\bm Z^{(k)})-\bm Z^{(k)\top}\bm P^{(k)\top}\hat{\bm c}^{(k)}\}
        \mid\mathcal Z)
        \\\notag
        &=
        n^{-1}
        {\rm E}\Big(
        {\rm tr}
        \Big[
        \{{\rm E}(\bm Y^{(k)\top}\mid\bm Z^{(k)})
        -\bm Y^{(k)\top}
        +\bm Y^{(k)\top}
        -\hat{\bm c}^{(k)\top}\bm P^{(k)}\bm Z^{(k)}\}
        \\\notag
        &\qquad
        \bm A_2^{(k)}
        \{{\rm E}(\bm Y^{(k)}\mid\bm Z^{(k)})
        -\bm Y^{(k)\top}
        +\bm Y^{(k)\top}
        -\bm Z^{(k)\top}\bm P^{(k)\top}\hat{\bm c}^{(k)}\}
        \Big]
        \mid\mathcal Z
        \Big)
        \\\notag
        &\leq
        2n^{-1}
        {\rm tr}
        \Big\{
        \bm A_2^{(k)}
        {\rm E}(\bm Y^{(k)2}\mid\bm Z^{(k)})
        \Big\}
        \\\notag
        &=
        2n^{-1}
        {\rm tr}(\bm A_2^{(k)})
        \sum_{i=1}^n
        \frac{a_{ii}(\bm A_2^{(k)})}{{\rm tr}(\bm A_2^{(k)})}
        {\rm E}(Y_i^{(k)2}\mid\bm Z_i^{(k)})
        \\\notag
        &=
        O_p(n^{-1}\lambda^{-1/4}J^\xi).
        \quad\text{(refer to \eqref{eq:tr_A2} and \citet[Theorem 1]{Pruitt1966})}
    \end{align}
    Analogously,
    we have
    \begin{align}
        \notag
        {\rm E}&(\|
        \bm\phi^\top\{
        \hat{\bm c}^{(k)}-
        \widehat{\rm E}(\hat{\bm c}^{(k)}\mid\bm Z^{(k)})
        \}
        \|_{\widehat C^{(0)}}^2
        \mid\mathcal Z)
        \\\notag
        &=
        {\rm E}(\|
        n^{-\frac{1}{2}}
        (n^{-\frac{1}{2}}\bm Z^{(0)\top}\bm P^{(0)\top})
        (\bm\Omega^{(k)}+\lambda\bm W)^{-1}
        (n^{-\frac{1}{2}}\bm P^{(k)}\bm Z^{(k)})
        (\hat{\bm c}^{(k)}-\bm Z^{(k)\top}\bm P^{(k)\top}\hat{\bm c}^{(k)})
        \|_2^2\mid\mathcal Z)
        \\\notag
        &=
        n^{-1}
        {\rm E}(
        (\hat{\bm c}^{(k)\top}-\hat{\bm c}^{(k)\top}\bm P^{(k)}\bm Z^{(k)})
        \bm A_2^{(k)}
        (\hat{\bm c}^{(k)}-\bm Z^{(k)\top}\bm P^{(k)\top}\hat{\bm c}^{(k)})
        \mid\mathcal Z)
        \\\notag
        &\leq
        n^{-1}
        {\rm tr}
        \Big\{
        \bm A_2^{(k)}
        {\rm E}(\bm Y^{(k)2}\mid\bm Z^{(k)})
        \Big\}
        \\\notag
        &=
        n^{-1}
        {\rm tr}(\bm A_2^{(k)})
        \sum_{i=1}^n
        \frac{a_{ii}(\bm A_2^{(k)})}{{\rm tr}(\bm A_2^{(k)})}
        {\rm E}(Y_i^{(k)2}\mid\bm Z_i^{(k)})
        \\\notag
        &=
        O_p(n^{-1}\lambda^{-1/4}J^\xi).
        \quad\text{(refer to \eqref{eq:tr_A2} and \citet[Theorem 1]{Pruitt1966})}
    \end{align}

\end{proof}

\begin{proof}[Proof of Proposition \ref{prop:hat_beta_c_estimation_err}]
    Observe that
    \begin{align*}
        \bm U^*&\{\hat{\bm\delta}-\widehat{\rm E}(\hat{\bm\delta}\mid\mathcal Z)\}
        \\
        &=
        \bigg\{\sum_{k=0}^K{\rm var}^{-1}(\hat{\bm c}^{(k)}\mid\bm Z^{(k)})\bigg\}^{-1}
        \\
        &\qquad
        \sum_{k=1}^K{\rm var}^{-1}(\hat{\bm c}^{(k)}\mid\bm Z^{(k)})
        \big[
        \hat{\bm c}^{(0)}-\widehat{\rm E}(\hat{\bm c}^{(0)}\mid\bm Z^{(0)})
        -
        \{\hat{\bm c}^{(k)}-\widehat{\rm E}(\hat{\bm c}^{(k)}\mid\bm Z^{(k)})\}
        \big].
    \end{align*}
    It follows from Lemma \ref{lemma:aUBUa} that,
    for positive semidefinite $\bm B\in\mathbb R^{M\times M}$, 
    \begin{align}
        \notag
        [\bm U^*&\{\hat{\bm\delta}-\widehat{\rm E}(\hat{\bm\delta}\mid\mathcal Z)\}]^\top
        \bm B
        \bm U^*\{\hat{\bm\delta}-\widehat{\rm E}(\hat{\bm\delta}\mid\mathcal Z)\}
        \\ \notag
        &\leq
        K
        \sum_{k=1}^K
        \bigg(
        \big[
        \hat{\bm c}^{(0)}-\widehat{\rm E}(\hat{\bm c}^{(0)}\mid\bm Z^{(0)})
        -
        \{\hat{\bm c}^{(k)}-\widehat{\rm E}(\hat{\bm c}^{(k)}\mid\bm Z^{(k)})\}
        \big]^\top
        \\ \label{eq:bound_UE1}
        &\qquad\qquad
        \bm B\big[
        \hat{\bm c}^{(0)}-\widehat{\rm E}(\hat{\bm c}^{(0)}\mid\bm Z^{(0)})
        -
        \{\hat{\bm c}^{(k)}-\widehat{\rm E}(\hat{\bm c}^{(k)}\mid\bm Z^{(k)})\}
        \big]
        \bigg),
    \end{align}
    \begin{align}
        \notag
        [\widehat{\bm U}^*&\{\hat{\bm\delta}-\widehat{\rm E}(\hat{\bm\delta}\mid\mathcal Z)\}]^\top
        \bm B
        \widehat{\bm U}^*\{\hat{\bm\delta}-\widehat{\rm E}(\hat{\bm\delta}\mid\mathcal Z)\}
        \\ \notag
        &\leq
        K\sum_{k=1}^K
        \bigg(
        \big[
        \hat{\bm c}^{(0)}-\widehat{\rm E}(\hat{\bm c}^{(0)}\mid\bm Z^{(0)})
        -
        \{\hat{\bm c}^{(k)}-\widehat{\rm E}(\hat{\bm c}^{(k)}\mid\bm Z^{(k)})\}
        \big]^\top
        \\ \label{eq:bound_UE2}
        &\qquad\qquad
        \bm B\big[
        \hat{\bm c}^{(0)}-\widehat{\rm E}(\hat{\bm c}^{(0)}\mid\bm Z^{(0)})
        -
        \{\hat{\bm c}^{(k)}-\widehat{\rm E}(\hat{\bm c}^{(k)}\mid\bm Z^{(k)})\}
        \big]
        \bigg),
    \end{align}
    and
    \begin{align}
        \notag 
        &\bm U^*\{{\rm E}(\hat{\bm\delta}\mid\mathcal Z)
        -\widehat{\rm E}(\hat{\bm\delta}\mid\mathcal Z)\}
        \\ \notag
        &\leq
        K\sum_{k=1}^K
        \bigg(
        \big[
        {\rm E}(\hat{\bm c}^{(0)}\mid\bm Z^{(0)})-\widehat{\rm E}(\hat{\bm c}^{(0)}\mid\bm Z^{(0)})
        -
        \{
        {\rm E}(\hat{\bm c}^{(k)}\mid\bm Z^{(k)})-\widehat{\rm E}(\hat{\bm c}^{(k)}\mid\bm Z^{(k)})\}
        \big]^\top
        \\ \label{eq:bound_UE3}
        &\qquad
        \bm B\big[
        {\rm E}(\hat{\bm c}^{(0)}\mid\bm Z^{(0)})-\widehat{\rm E}(\hat{\bm c}^{(0)}\mid\bm Z^{(0)})
        -
        \{
        {\rm E}(\hat{\bm c}^{(k)}\mid\bm Z^{(k)})-\widehat{\rm E}(\hat{\bm c}^{(k)}\mid\bm Z^{(k)})\}
        \big]
        \bigg).
    \end{align}
    Applying the above three bounds,
    \begin{align*}
        &{\rm E}(\|\hat\beta_C^{(0)}-\tilde\beta^{(0)}\|_{\widehat C^{(0)}}^2\mid\mathcal Z)
        \\
        &=
        {\rm E}
        (\|
        \bm\phi^\top(\bm U^*-\widehat{\bm U}^*)
        \{\hat{\bm\delta}-\widehat{\rm E}(\hat{\bm\delta}\mid\mathcal Z)\}
        -
        \bm\phi^\top\bm U^*
        \{{\rm E}(\hat{\bm\delta}\mid\mathcal Z)-\widehat{\rm E}(\hat{\bm\delta}\mid\mathcal Z)\}
        \|_{\widehat C^{(0)}}^2
        \mid\mathcal Z
        )
        \\
        &\leq
        3{\rm E}
        (\|
        \bm\phi^\top\bm U^*
        \{\hat{\bm\delta}-\widehat{\rm E}(\hat{\bm\delta}\mid\mathcal Z)\}
        \|_{\widehat C^{(0)}}^2
        \mid\mathcal Z
        )
        +
        \\
        &\qquad
        3{\rm E}
        (\|
        \bm\phi^\top\widehat{\bm U}^*
        \{\hat{\bm\delta}-\widehat{\rm E}(\hat{\bm\delta}\mid\mathcal Z)\}
        \|_{\widehat C^{(0)}}^2
        \mid\mathcal Z
        )
        +
        \\
        &\qquad
        3{\rm E}
        (\|
        \bm\phi^\top{\bm U}^*
        \{{\rm E}(\hat{\bm\delta}\mid\mathcal Z)-\widehat{\rm E}(\hat{\bm\delta}\mid\mathcal Z)\}
        \|_{\widehat C^{(0)}}^2
        \mid\mathcal Z
        )
        \\
        &\leq
        3K\Big\{
        {\rm E}
        (\|
        \bm\phi^\top
        [
        \{\hat{\bm c}^{(0)}-\widehat{\rm E}(\hat{\bm c}^{(0)}\mid\mathcal Z)\}
        -\{\hat{\bm c}^{(k)}-\widehat{\rm E}(\hat{\bm c}^{(k)}\mid\mathcal Z)\}
        ]
        \|_{\widehat C^{(0)}}^2
        \mid\mathcal Z
        )
        +
        \quad\text{(refer to \eqref{eq:bound_UE1})}
        \\
        &\qquad
        {\rm E}
        (\|
        \bm\phi^\top
        [
        \{\hat{\bm c}^{(0)}-\widehat{\rm E}(\hat{\bm c}^{(0)}\mid\mathcal Z)\}
        -\{\hat{\bm c}^{(k)}-\widehat{\rm E}(\hat{\bm c}^{(k)}\mid\mathcal Z)\}
        ]
        \|_{\widehat C^{(0)}}^2
        \mid\mathcal Z
        )
        +
        \quad\text{(refer to \eqref{eq:bound_UE2})}
        \\
        &\qquad
        {\rm E}
        (
        \|
        \bm\phi^\top
        [
        \{{\rm E}(\hat{\bm c}^{(0)}\mid\mathcal Z)
        -\widehat{\rm E}(\hat{\bm c}^{(0)}\mid\mathcal Z)\}
        -\{{\rm E}(\hat{\bm c}^{(k)}\mid\mathcal Z)
        -\widehat{\rm E}(\hat{\bm c}^{(k)}\mid\mathcal Z)\}
        ]
        \|_{\widehat C^{(0)}}^2
        \mid\mathcal Z
        )
        \Big\}
        \quad\text{(see \eqref{eq:bound_UE3})}
        \\
        &=
        O_p(n^{-1}\lambda^{-1/4}J^\xi).
        \quad\text{(Lemma \ref{lemma:bounds_c_Ec})}
    \end{align*}
    Lemmas \ref{lemma:pred_err_internal} and \ref{lemma:bounds_c_Ec} help finalize the proof:
    \begin{align*}
        {\rm E}&(\|\hat\beta_C^{(0)}-\beta^{(0)}\|_{\widehat C^{(0)}}^2\mid\mathcal Z)
        \\
        &\leq
        2{\rm E}(\|\hat\beta^{(0)}-\beta^{(0)}\|_{\widehat C^{(0)}}^2\mid\mathcal Z)
        +
        2{\rm E}(
        \|\bm\phi^\top\widehat{\bm U}^*
        \{\hat{\bm\delta}-\widehat{\rm E}(\hat{\bm\delta}\mid\mathcal Z)\}\|_{\widehat C^{(0)}}^2\mid\mathcal Z)
        \\
        &\leq
        2{\rm E}(\|\hat\beta^{(0)}-\beta^{(0)}\|_{\widehat C^{(0)}}^2\mid\mathcal Z)
        +
        \\
        &\qquad
        2{\rm E}
        (\|
        \bm\phi^\top
        [
        \{\hat{\bm c}^{(0)}-\widehat{\rm E}(\hat{\bm c}^{(0)}\mid\mathcal Z)\}
        -\{\hat{\bm c}^{(k)}-\widehat{\rm E}(\hat{\bm c}^{(k)}\mid\mathcal Z)\}
        ]
        \|_{\widehat C^{(0)}}^2
        \mid\mathcal Z
        )
        \quad\text{(see \eqref{eq:bound_UE2})}
        \\
        &=
        O_p(\lambda+\rho+J^{-1}\rho^{-1/4}+n^{-1}\lambda^{-1/4}J^\xi).
        \quad\text{(Lemmas \ref{lemma:pred_err_internal} and \ref{lemma:bounds_c_Ec})}
    \end{align*}
\end{proof}

\begin{proof}[Proof of Proposition \ref{prop:hat_beta_c_pred_err}]
    The simple algebra gives that
    \begin{align*}
        {\rm E}&\left(
        \frac{1}{n}\sum_{i=1}^n
        \langle\tilde\beta^{(0)}-\hat\beta_C^{(0)},\widehat X_i^{(00)}\rangle_{L^2}^2
        \mid\mathcal Z, \bm Z^{(00)}
        \right)
        \\
        &=
        {\rm E}
        (\|
        (n^{-\frac{1}{2}}\bm Z^{(00)\top}\bm P^{(00)\top})
        [
        \bm U^*
        \{\hat{\bm\delta}-{\rm E}(\hat{\bm\delta}\mid\mathcal Z)\}
        -
        \widehat{\bm U}^*
        \{\hat{\bm\delta}-\widehat{\rm E}(\hat{\bm\delta}\mid\mathcal Z)\}
        ]
        \|_2^2
        \mid\mathcal Z, \bm Z^{(00)}
        )
        \\
        &=
        {\rm E}
        \Big(
        [
        \bm U^*
        \{\hat{\bm\delta}-{\rm E}(\hat{\bm\delta}\mid\mathcal Z))\}
        -
        \widehat{\bm U}^*
        \{\hat{\bm\delta}-\widehat{\rm E}(\hat{\bm\delta}\mid\mathcal Z))\}
        ]^\top
        \bm\Omega^{(0)\frac{1}{2}}
        (\bm\Omega^{(0)-\frac{1}{2}}\bm\Omega^{(00)}\bm\Omega^{(0)-\frac{1}{2}})
        \\
        &\qquad
        \bm\Omega^{(0)\frac{1}{2}}
        [
        \bm U^*
        \{\hat{\bm\delta}-{\rm E}(\hat{\bm\delta}\mid\mathcal Z))\}
        -
        \widehat{\bm U}^*
        \{\hat{\bm\delta}-\widehat{\rm E}(\hat{\bm\delta}\mid\mathcal Z))\}
        ]
        \mid\mathcal Z, \bm Z^{(00)})
        \Big)
        \\
        &\leq
        \gamma_1
        (\bm\Omega^{(0)-\frac{1}{2}}\bm\Omega^{(00)}\bm\Omega^{(0)-\frac{1}{2}})
        {\rm E}
        \Big(
        [
        \bm U^*
        \{\hat{\bm\delta}-{\rm E}(\hat{\bm\delta}\mid\mathcal Z)\}
        -
        \widehat{\bm U}^*
        \{\hat{\bm\delta}-\widehat{\rm E}(\hat{\bm\delta}\mid\mathcal Z)\}
        ]^\top
        \\
        &\qquad
        \bm\Omega^{(0)}
        [
        \bm U^*
        \{\hat{\bm\delta}-{\rm E}(\hat{\bm\delta}\mid\mathcal Z)\}
        -
        \widehat{\bm U}^*
        \{\hat{\bm\delta}-\widehat{\rm E}(\hat{\bm\delta}\mid\mathcal Z)\}
        ]
        \mid\mathcal Z
        \Big)
        \\
        &=
        \bigg(\max_{\bm a\neq\bm 0}
        \frac{\bm a^\top\bm\Omega^{(00)}\bm a}{\bm a^\top\bm\Omega^{(0)}\bm a}
        \bigg)
        {\rm E}(\|\hat\beta_C^{(0)}-\tilde\beta^{(0)}\|_{\widehat C^{(0)}}^2\mid\mathcal Z)
        \\
        &=O_p(n^{-1}\lambda^{-1/4}J^\xi)
        \quad\text{(Proposition \ref{prop:hat_beta_c_estimation_err} \& Lemma \ref{lemma:bound_X})}
    \end{align*}
    and
    \begin{align*}
        \frac{1}{n}&\sum_{i=1}^n
        \left\{
        \langle\beta^{(0)},X_i^{(00)}\rangle_{L^2}-
        \langle\hat\beta_C^{(0)},\widehat X_i^{(00)}\rangle_{L^2}
        \right\}^2
        \\
        &=
        \frac{1}{n}\sum_{i=1}^n
        \left\{
        \langle\beta^{(0)},X_i^{(00)}-\widehat X_i^{(00)}\rangle_{L^2}-
        \langle\hat\beta_C^{(0)}-\hat\beta^{(0)},\widehat X_i^{(00)}\rangle_{L^2}+
        \langle\beta^{(0)}-\hat\beta^{(0)},\widehat X_i^{(00)}\rangle_{L^2}
        \right\}^2
        \\
        &\leq
        \frac{3}{n}\sum_{i=1}^n
        \langle\beta^{(0)},X_i^{(00)}-\widehat X_i^{(00)}\rangle_{L^2}^2
        +
        \frac{3}{n}\sum_{i=1}^n
        \langle\hat\beta^{(0)}-\hat\beta_C^{(0)},\widehat X_i^{(00)}\rangle_{L^2}^2
        +
        \frac{3}{n}\sum_{i=1}^n
        \langle\beta^{(0)}-\hat\beta^{(0)},\widehat X_i^{(00)}\rangle_{L^2}^2
        \\
        &\leq
        \|\beta^{(0)}\|_{L^2}^2
        \frac{3}{n}\sum_{i=1}^n
        \|X_i^{(00)}-\widehat X_i^{(00)}\|_{L^2}^2
        +
        \frac{3}{n}\sum_{i=1}^n
        \langle\hat\beta^{(0)}-\hat\beta_C^{(0)},\widehat X_i^{(00)}\rangle_{L^2}^2
        +
        \\
        &\qquad
        \frac{3}{n}\sum_{i=1}^n
        \langle\beta^{(0)}-\hat\beta^{(0)},\widehat X_i^{(00)}\rangle_{L^2}^2.
    \end{align*}
    The proof is completed by combining
    the Cauchy-Schwarz inequality,
    Lemmas \ref{lemma:bound_X} and \ref{lemma:pred_err_external}
    with the following identity
    \begin{align*}
        {\rm E}&\left(
        \frac{1}{n}\sum_{i=1}^n
        \langle\hat\beta^{(0)}-\hat\beta_C^{(0)},\widehat X_i^{(00)}\rangle_{L^2}^2
        \mid\mathcal Z,\bm Z^{(00)}
        \right)
        \\
        &=
        {\rm E}
        (\|
        (n^{-\frac{1}{2}}\bm Z^{(00)\top}\bm P^{(00)\top})
        \widehat{\bm U}^*
        \{\hat{\bm\delta}-\widehat{\rm E}(\hat{\bm\delta}\mid\mathcal Z)\}
        \|_2^2
        \mid\mathcal Z, \bm Z^{(00)}
        )
        \\
        &=
        {\rm E}
        \Big(
        [
        \widehat{\bm U}^*
        \{\hat{\bm\delta}-\widehat{\rm E}(\hat{\bm\delta}\mid\mathcal Z)\}
        ]^\top
        \bm\Omega^{(0)\frac{1}{2}}
        (\bm\Omega^{(0)-\frac{1}{2}}\bm\Omega^{(00)}\bm\Omega^{(0)-\frac{1}{2}})
        \bm\Omega^{(0)\frac{1}{2}}
        [
        \widehat{\bm U}^*
        \{\hat{\bm\delta}-\widehat{\rm E}(\hat{\bm\delta}\mid\mathcal Z)\}
        ]
        \mid\mathcal Z, \bm Z^{(00)}
        \Big)
        \\
        &\leq
        \gamma_1
        (\bm\Omega^{(0)-\frac{1}{2}}\bm\Omega^{(00)}\bm\Omega^{(0)-\frac{1}{2}})
        {\rm E}
        \Big(
        [
        \widehat{\bm U}^*
        \{\hat{\bm\delta}-\widehat{\rm E}(\hat{\bm\delta}\mid\mathcal Z)\}
        ]^\top
        \bm\Omega^{(0)}
        [
        \widehat{\bm U}^*
        \{\hat{\bm\delta}-\widehat{\rm E}(\hat{\bm\delta}\mid\mathcal Z)\}
        ]
        \mid\mathcal Z
        \Big)
        \\
        &=
        2\bigg(\max_{\bm a\neq\bm 0}
        \frac{\bm a^\top\bm\Omega^{(00)}\bm a}{\bm a^\top\bm\Omega^{(0)}\bm a}
        \bigg)
        {\rm E}
        \big(
        \\
        &\qquad
        \|
        \bm\phi^\top
        [
        \{\hat{\bm c}^{(0)}
        -\widehat{\rm E}(\hat{\bm c}^{(0)}\mid\mathcal Z)\}
        -\{\hat{\bm c}^{(k)}
        -\widehat{\rm E}(\hat{\bm c}^{(k)}\mid\mathcal Z)\}
        ]
        \|_{\widehat C^{(0)}}^2
        \mid\mathcal Z
        \big)
        \quad\text{(refer to \eqref{eq:bound_UE2})}
        \\
        &=O_p(n^{-1}\lambda^{-1/4}J^\xi).
        \quad\text{(Lemma \ref{lemma:bounds_c_Ec})}
    \end{align*}
\end{proof}

\begin{proof}[Proof of Proposition \ref{prop:hat_beta_pc_err}]
    First,
    it suffices to bound
    $
    {\rm E}
        (\{
        \gamma_1(\bm\Omega^{(0)\frac{1}{2}}\widehat{\rm var}(\hat{\bm c}^{(k)}\mid\bm Z^{(k)}))
        \}^2
        \mid\mathcal Z)
    $
    from above.
    By the Cauchy-Schwarz inequality,
    \begin{align*}
        \hat\sigma_{\varepsilon^{(k)}}^4
        &\leq
        \left\{
        \frac{3}{n}
        \sum_{i=1}^{n}
        (\varepsilon_i^{(k)2}+
        \langle X_i^{(k)}-\widehat X_i^{(k)},\beta^{(k)}\rangle_{L^2}^2+
        \langle \widehat X_i^{(k)},\beta^{(k)}-\hat\beta^{(k)}\rangle_{L^2}^2
        )
        \right\}^2
        \\
        &\leq
        3\left\{
        \frac{3}{n}
        \sum_{i=1}^{n}
        \varepsilon_i^{(k)2}
        \right\}^2
        +
        3\left\{
        \frac{3}{n}
        \sum_{i=1}^{n}
        \langle X_i^{(k)}-\widehat X_i^{(k)},\beta^{(k)}\rangle_{L^2}^2
        \right\}^2
        +
        \\
        &\qquad
        3\left\{
        \frac{3}{n}
        \sum_{i=1}^{n}
        \langle \widehat X_i^{(k)},\beta^{(k)}-\hat\beta^{(k)}\rangle_{L^2}^2
        \right\}^2
        \\
        &\leq
        \frac{27}{n}
        \sum_{i=1}^{n}
        \varepsilon_i^{(k)4}
        +
        \frac{27}{n}
        \|\beta^{(k)}\|_{L^2}^4
        \sum_{i=1}^{n}
        \|X_i^{(k)}-\widehat X_i^{(k)}\|_{L^2}^4
        +
        27\left\{
        \frac{1}{n}
        \sum_{i=1}^{n}
        \langle \widehat X_i^{(k)},\beta^{(k)}-\hat\beta^{(k)}\rangle_{L^2}^2
        \right\}^2,
    \end{align*}
    which is followed by 
    ${\rm E}(\hat\sigma_{\varepsilon^{(k)}}^4\mid\mathcal Z)=O_p(1)$;
    see Lemmas \ref{lemma:bound_X} and \ref{lemma:pred_err_internal}.
    Furthermore,
    \begin{align}
        \notag
        &{\rm E}
        (\{
        \gamma_1(\bm\Omega^{(0)\frac{1}{2}}\widehat{\rm var}(\hat{\bm c}^{(k)}\mid\bm Z^{(k)}))
        \}^2
        \mid\mathcal Z)
        \\ \notag
        &= 
        \frac{{\rm E}(\hat\sigma_{\varepsilon^{(k)}}^4\mid\mathcal Z)}{n^2}
        \gamma_1(
        (\bm\Omega^{(k)}+\lambda\bm W)^{-1}
        \bm\Omega^{(k)}
        (\bm\Omega^{(k)}+\lambda\bm W)^{-1}
        \bm\Omega^{(0)}
        (\bm\Omega^{(k)}+\lambda\bm W)^{-1}
        \bm\Omega^{(k)}(\bm\Omega^{(k)}+\lambda\bm W)^{-1})
        \\ \notag
        &\leq
        \frac{{\rm E}(\hat\sigma_{\varepsilon^{(k)}}^4\mid\mathcal Z)}{n^2}
        \gamma_1(
        (\bm\Omega^{(k)}+\lambda\bm W)^{-1}
        \bm\Omega^{(k)\frac{1}{2}}
        (\bm\Omega^{(k)-\frac{1}{2}}
        \bm\Omega^{(0)}
        \bm\Omega^{(k)-\frac{1}{2}})
        \bm\Omega^{(k)\frac{1}{2}}
        (\bm\Omega^{(k)}+\lambda\bm W)^{-1}
        )
        \\ \notag
        &\leq
        \frac{{\rm E}(\hat\sigma_{\varepsilon^{(k)}}^4\mid\mathcal Z)}{n^2}
        \gamma_1(
        \bm\Omega^{(k)-\frac{1}{2}}
        \bm\Omega^{(0)}
        \bm\Omega^{(k)-\frac{1}{2}}
        )
        \gamma_1(
        (\bm\Omega^{(k)}+\lambda\bm W)^{-1}
        )
        \quad\text{(Lemma \ref{lemma:eigen_inequality})}
        \\ \label{eq:max_eigen_omega_half_var_ck}
        &=
        O_p(n^{-2}\lambda^{-1}J^\xi),
    \end{align}
    noting $\gamma_1((\bm\Omega^{(k)}+\lambda\bm W)^{-1})=O_p(\lambda^{-1})$ \citep[pp.~1803]{LiHsing2007}.
    Given $K_0\in[1, K-1]$,
    without loss of generality,
    suppose  
    $\bm\delta^{*(k)}\in\mathbb R^M$
    is a zero vector for each $k > K_0$
    and is nonzero otherwise.
    Partition $\hat{\bm\delta}$ into two parts,
    $\hat{\bm\delta}_1
    =[(\hat{\bm c}^{(0)}-\hat{\bm c}^{(1)})^\top,\ldots,
    (\hat{\bm c}^{(0)}-\hat{\bm c}^{(K_0)})^\top
    ]^\top
    $
    and
    $\hat{\bm\delta}_2
    =[(\hat{\bm c}^{(0)}-\hat{\bm c}^{(K_0+1)})^\top,\ldots,
    (\hat{\bm c}^{(0)}-\hat{\bm c}^{(K)})^\top
    ]^\top
    $.
    Applying the exact partition to $\bm{\delta}^\zeta$,
    define 
    $\bm\delta_1^*=[\bm\delta^{*(1)\top},\ldots,\bm\delta^{*(K_0)\top}]^\top$
    and
    $\bm\delta_2^*=[\bm\delta^{*(K_0+1)\top},\ldots,\bm\delta^{*(K)\top}]^\top$.
    The Karush-Kuhn-Tucker conditions give that
    \begin{align}
        \label{eq:kkt1}
        \bm0 
        &=
        -2
        \bm A_{11}
        (\hat{\bm\delta}_1-\bm\delta_1^*)
        -
        2
        \bm A_{12}
        \hat{\bm\delta}_2
        +
        \zeta
        \left[
        \frac{\bm\delta^{*(1)\top}}{\|\bm\delta_1^{*(1)}\|_2},\ldots,
        \frac{\bm\delta^{*(K_0)\top}}{\|\bm\delta_1^{*(K_0)}\|_2}
        \right]^\top
        ,
        \\ \label{eq:kkt2}
        \bm0 
        &=
        -2
        \bm A_{22}
        \hat{\bm\delta}_2
        -
        2
        \bm A_{21}
        (\hat{\bm\delta}_1-\bm\delta_1^*)
        +
        \zeta
        \left[
        \bm u_1^\top,\ldots,
        \bm u_{K-K_0}^\top
        \right]^\top,
    \end{align}
    for certain $\bm u_1,\ldots,\bm u_{K-K_0}\in\{\bm u\in\mathbb R^M:\|\bm u\|_2\leq 1\}$,
    where 
    \begin{align*}
        \bm A_{11}
        &=
        {\rm diag}\{
        \widehat{\rm var}^{-1}(\hat{\bm c}^{(1)}\mid\bm Z^{(1)}),\ldots,
        \widehat{\rm var}^{-1}(\hat{\bm c}^{(K_0)}\mid\bm Z^{(K_0)})
        \}
        -
        \\
        &\qquad
        [
        \widehat{\rm var}^{-1}(\hat{\bm c}^{(1)}\mid\bm Z^{(1)}),
        \ldots,
        \widehat{\rm var}^{-1}(\hat{\bm c}^{(K_0)}\mid\bm Z^{(K_0)})
        ]^\top
        \\
        &\qquad
        \bigg\{
        \sum_{k=0}^K
        \widehat{\rm var}^{-1}(\hat{\bm c}^{(k)}\mid\bm Z^{(k)})
        \bigg\}^{-1}
        [
        \widehat{\rm var}^{-1}(\hat{\bm c}^{(1)}\mid\bm Z^{(1)}),
        \ldots,
        \widehat{\rm var}^{-1}(\hat{\bm c}^{(K_0)}\mid\bm Z^{(K_0)})
        ],
    \end{align*}
    \begin{align*}
        \bm A_{12}
        &=
        \bm A_{21}^\top
        =
        -[
        \widehat{\rm var}^{-1}(\hat{\bm c}^{(1)}\mid\bm Z^{(1)}),
        \ldots,
        \widehat{\rm var}^{-1}(\hat{\bm c}^{(K_0)}\mid\bm Z^{(K_0)})
        ]^\top
        \\
        &\qquad
        \bigg\{
        \sum_{k=0}^K
        \widehat{\rm var}^{-1}(\hat{\bm c}^{(k)}\mid\bm Z^{(k)})
        \bigg\}^{-1}
        [
        \widehat{\rm var}^{-1}(\hat{\bm c}^{(K_0+1)}\mid\bm Z^{(K_0+1)}),
        \ldots,
        \widehat{\rm var}^{-1}(\hat{\bm c}^{(K)}\mid\bm Z^{(K)})
        ],
    \end{align*}
    and
    \begin{align*}
        \bm A_{22}
        &=
        {\rm diag}\{
        \widehat{\rm var}^{-1}(\hat{\bm c}^{(K_0+1)}\mid\bm Z^{(K_0+1)}),\ldots,
        \widehat{\rm var}^{-1}(\hat{\bm c}^{(K)}\mid\bm Z^{(K)})
        \}
        -
        \\
        &\qquad
        [
        \widehat{\rm var}^{-1}(\hat{\bm c}^{(K_0+1)}\mid\bm Z^{(K_0+1)}),
        \ldots,
        \widehat{\rm var}^{-1}(\hat{\bm c}^{(K)}\mid\bm Z^{(K)})
        ]^\top
        \\
        &\qquad
        \bigg\{
        \sum_{k=0}^K
        \widehat{\rm var}^{-1}(\hat{\bm c}^{(k)}\mid\bm Z^{(k)})
        \bigg\}^{-1}
        [
        \widehat{\rm var}^{-1}(\hat{\bm c}^{(K_0+1)}\mid\bm Z^{(K_0+1)}),
        \ldots,
        \widehat{\rm var}^{-1}(\hat{\bm c}^{(K)}\mid\bm Z^{(K)})
        ].
    \end{align*}
    By the Woodbury's matrix identity 
    \citep[see, e.g.,][Theorem 1.9]{Schott2017},
    \begin{align*}
        \bm A_{11}^{-1}
        &=
        {\rm diag}\{
        \widehat{\rm var}(\hat{\bm c}^{(1)}\mid\bm Z^{(1)}),\ldots,
        \widehat{\rm var}(\hat{\bm c}^{(K_0)}\mid\bm Z^{(K_0)})
        \}
        +
        \\
        &\qquad
        (\bm 1_{K_0}\bm 1_{K_0}^\top)
        \otimes
        \bigg\{
        \sum_{k\in\{0,K_0+1,\ldots,K\}}
        \widehat{\rm var}^{-1}(\hat{\bm c}^{(k)}\mid\bm Z^{(k)})
        \bigg\}^{-1},
    \end{align*}
    \begin{align*}
        \bm A_{22}^{-1}
        &=
        {\rm diag}\{
        \widehat{\rm var}(\hat{\bm c}^{(K_0+1)}\mid\bm Z^{(K_0+1)}),\ldots,
        \widehat{\rm var}(\hat{\bm c}^{(K)}\mid\bm Z^{(K)})
        \}
        +
        \\
        &\qquad
        (\bm 1_{K-K_0}\bm 1_{K-K_0}^\top)
        \otimes
        \bigg\{
        \sum_{k\in\{0,\ldots,K_0\}}
        \widehat{\rm var}^{-1}(\hat{\bm c}^{(k)}\mid\bm Z^{(k)})
        \bigg\}^{-1},
    \end{align*}
    \begin{align*}
        (&\bm A_{11}-\bm A_{12}\bm A_{22}^{-1}\bm A_{21})^{-1}
        \\
        &=
        (\bm 1_{K_0}\bm 1_{K_0}^\top)
        \otimes
        \widehat{\rm var}(\hat{\bm c}^{(0)}\mid\bm Z^{(0)})
        +
        {\rm diag}\{
        \widehat{\rm var}(\hat{\bm c}^{(1)}\mid\bm Z^{(1)}),\ldots,
        \widehat{\rm var}(\hat{\bm c}^{(K_0)}\mid\bm Z^{(K_0)})
        \},
    \end{align*}
    and
    \begin{align*}
        (&\bm A_{22}-\bm A_{21}\bm A_{11}^{-1}\bm A_{12})^{-1}
        \\
        &=
        (\bm 1_{K-K_0}\bm 1_{K-K_0}^\top)
        \otimes
        \widehat{\rm var}(\hat{\bm c}^{(0)}\mid\bm Z^{(0)})
        +
        {\rm diag}\{
        \widehat{\rm var}(\hat{\bm c}^{(K_0+1)}\mid\bm Z^{(K_0+1)}),\ldots,
        \widehat{\rm var}(\hat{\bm c}^{(K)}\mid\bm Z^{(K)})
        \}.
    \end{align*}
    Solving \eqref{eq:kkt1} and \eqref{eq:kkt2} 
    results in
    \begin{align*}
        \hat{\bm\delta}_1&-\bm\delta_1^*
        =
        \\
        &
        \frac{\zeta}{2}
        (\bm A_{11}-\bm A_{12}\bm A_{22}^{-1}\bm A_{21})^{-1}
        \Bigg(
        \left[
        \frac{\bm\delta_1^{*(1)\top}}{\|\bm\delta_1^{*(1)}\|_2},
        \ldots,
        \frac{\bm\delta_1^{*(K_0)\top}}{\|\bm\delta_1^{*(K_0)}\|_2}
        \right]^\top
        -
        \bm A_{12}\bm A_{22}^{-1}
        \left[
        \bm u_1^\top,\ldots,
        \bm u_{K-K_0}^\top
        \right]^\top
        \Bigg),
        \\
        \hat{\bm\delta}_2&
        =
        \\
        &
        \frac{\zeta}{2}
        (\bm A_{22}-\bm A_{21}\bm A_{11}^{-1}\bm A_{12})^{-1}
        \Bigg(
        \left[
        \bm u_1^\top,\ldots,
        \bm u_{K-K_0}^\top
        \right]^\top
        -
        \bm A_{21}\bm A_{11}^{-1}
        \left[
        \frac{\bm\delta_1^{*(1)\top}}{\|\bm\delta_1^{*(1)}\|_2},
        \ldots,
        \frac{\bm\delta_1^{*(K_0)\top}}{\|\bm\delta_1^{*(K_0)}\|_2}
        \right]^\top
        \Bigg).
    \end{align*}
    This solution is followed by
    \begin{align*}
        &{\rm E}
        (\|(\bm I_{K_0}\otimes\bm\Omega^{(0)\frac{1}{2}})(\hat{\bm\delta}_1-\bm\delta_1^*)\|_2^2
        \mid\mathcal Z)
        \\
        &\leq
        \frac{\zeta^2K_0}{2}
        {\rm E}
        \big(
        \gamma_1
        \big(
        (\bm A_{11}-\bm A_{12}\bm A_{22}^{-1}\bm A_{21})^{-1}
        (\bm I_{K_0}\otimes\bm\Omega^{(0)})
        (\bm A_{11}-\bm A_{12}\bm A_{22}^{-1}\bm A_{21})^{-1}
        \big)
        \mid\mathcal Z
        \big)
        +
        \\
        &\qquad
        \frac{\zeta^2K_0}{2}
        {\rm E}
        \big(
        \gamma_1
        \big(
        \bm A_{22}^{-1}
        \bm A_{12}^\top
        (\bm A_{11}-\bm A_{12}\bm A_{22}^{-1}\bm A_{21})^{-1}
        (\bm I_{K_0}\otimes\bm\Omega^{(0)})
        (\bm A_{11}-\bm A_{12}\bm A_{22}^{-1}\bm A_{21})^{-1}
        \bm A_{12}\bm A_{22}^{-1}
        \big)
        \mid\mathcal Z
        \big)
        \\
        &\leq
        \frac{\zeta^2K_0}{2}
        {\rm E}
        \big(
        \{\gamma_1(\bm\Omega^{(0)\frac{1}{2}}{\rm var}(\hat{\bm c}^{(0)}\mid\bm Z^{(0)}))+
        \max_{1\leq k\leq K_0}\gamma_1(\bm\Omega^{(0)\frac{1}{2}}{\rm var}(\hat{\bm c}^{(k)}\mid\bm Z^{(k)}))
        \}^2
        \mid\mathcal Z
        \big)
        +
        \\
        &\qquad
        \frac{\zeta^2K_0}{2}
        {\rm E}
        \big(
        \gamma_1
        \big(
        \{(\bm 1_{K_0}\bm 1_{K_0}^\top)
        \otimes
        \widehat{\rm var}(\hat{\bm c}^{(0)}\mid\bm Z^{(0)})
        \}
        (\bm I_{K_0}\otimes\bm\Omega^{(0)})
        \{(\bm 1_{K_0}\bm 1_{K_0}^\top)
        \otimes
        \widehat{\rm var}(\hat{\bm c}^{(0)}\mid\bm Z^{(0)})
        \}
        \big)
        \mid\mathcal Z
        \big)
        \\
        &=
        \frac{\zeta^2K_0}{2}
        {\rm E}
        \big(
        \{\gamma_1(\bm\Omega^{(0)\frac{1}{2}}{\rm var}(\hat{\bm c}^{(0)}\mid\bm Z^{(0)}))+
        \max_{1\leq k\leq K_0}\gamma_1(\bm\Omega^{(0)\frac{1}{2}}{\rm var}(\hat{\bm c}^{(k)}\mid\bm Z^{(k)}))
        \}^2
        \mid\mathcal Z
        \big)
        +
        \\
        &\qquad
        \frac{\zeta^2K_0}{2}
        {\rm E}
        \big(
        \big\{
        \gamma_1
        \big(
        \bm\Omega^{(0)\frac{1}{2}}
        \widehat{\rm var}(\hat{\bm c}^{(0)}\mid\bm Z^{(0)})
        \big)
        \big\}^2
        \mid\mathcal Z
        \big)
        \\ 
        &=
        O_p(n^{-2}\lambda^{-1}J^\xi\zeta^2).
        \quad\text{(see \eqref{eq:max_eigen_omega_half_var_ck})}
    \end{align*}
    In the similar manner,
    \begin{align*}
        &{\rm E}
        (\|(\bm I_{K-K_0}\otimes\bm\Omega^{(0)\frac{1}{2}})\hat{\bm\delta}_2\|_2^2
        \mid\mathcal Z)
        \\
        &\leq
        \frac{\zeta^2K_0}{2}
        {\rm E}
        \big(
        \{\gamma_1(\bm\Omega^{(0)\frac{1}{2}}{\rm var}(\hat{\bm c}^{(0)}\mid\bm Z^{(0)}))+
        \max_{K_0+1\leq k\leq K}\gamma_1(\bm\Omega^{(0)\frac{1}{2}}{\rm var}(\hat{\bm c}^{(k)}\mid\bm Z^{(k)}))
        \}^2
        \mid\mathcal Z
        \big)
        +
        \\
        &\qquad
        \frac{\zeta^2K_0}{2}
        {\rm E}
        \big(
        \big\{
        \gamma_1
        \big(
        \bm\Omega^{(0)\frac{1}{2}}
        \widehat{\rm var}(\hat{\bm c}^{(0)}\mid\bm Z^{(0)})
        \big)
        \big\}^2
        \mid\mathcal Z
        \big)
        \\ 
        &=
        O_p(n^{-2}\lambda^{-1}J^\xi\zeta^2).
        \quad\text{(see \eqref{eq:max_eigen_omega_half_var_ck})}
    \end{align*}
    Taken together,
    the above two upper bounds imply that
    \begin{align*}
        {\rm E}&
        (\|(\bm I_K\otimes\bm\Omega^{(0)\frac{1}{2}})(\hat{\bm\delta}-\bm{\delta}^\zeta)\|_2^2
        \mid\mathcal Z
        )
        \\
        &=
        {\rm E}
        (\|(\bm I_{K_0}\otimes\bm\Omega^{(0)\frac{1}{2}})(\hat{\bm\delta}_1-\bm\delta_1^*)\|_2^2
        \mid\mathcal Z)
        +
        {\rm E}
        (\|(\bm I_{K-K_0}\otimes\bm\Omega^{(0)\frac{1}{2}})\hat{\bm\delta}_2\|_2^2
        \mid\mathcal Z)
        \\
        &=
        O_p(n^{-2}\lambda^{-1}\zeta^2J^\xi),
    \end{align*}
    and
    \begin{align*}
        {\rm E}&
        (\|(\bm I_K\otimes\bm\Omega^{(0)\frac{1}{2}})
        (\rm E(\hat{\bm\delta}\mid\mathcal Z)-\bm{\delta}^\zeta)\|_2^2
        \mid\mathcal Z
        )
        \\
        &\leq
        2{\rm E}
        (\|(\bm I_K\otimes\bm\Omega^{(0)\frac{1}{2}})(\hat{\bm\delta}-\bm{\delta}^\zeta)\|_2^2
        \mid\mathcal Z)
        +
        2{\rm E}
        (\|(\bm I_K\otimes\bm\Omega^{(0)\frac{1}{2}})
        (\rm E(\hat{\bm\delta}\mid\mathcal Z)-\hat{\bm\delta})
        \|_2^2
        \mid\mathcal Z)
        \\
        &=
        O_p(n^{-1}\lambda^{-1/4}J^\xi+n^{-2}\lambda^{-1}\zeta^2J^\xi).
        \quad\text{(Lemma \ref{lemma:bounds_c_Ec})}
    \end{align*}
    
    It follows from
    \begin{align*}
        \bm U^*(\hat{\bm\delta}-\bm{\delta}^\zeta)
        =
        \bigg\{\sum_{k=0}^K{\rm var}^{-1}(\hat{\bm c}^{(k)}\mid\bm Z^{(k)})\bigg\}^{-1}
        \sum_{k=1}^K{\rm var}^{-1}(\hat{\bm c}^{(k)}\mid\bm Z^{(k)})
        (\hat{\bm c}^{(0)}-\hat{\bm c}^{(k)}-\bm\delta^{*(k)})
    \end{align*}
    and Lemma \ref{lemma:aUBUa} that,
    for positive definite $\bm B\in\mathbb R^{M\times M}$, 
    \begin{align}
        \notag
        \{\bm U^*&(\hat{\bm\delta}-\bm{\delta}^\zeta)\}^\top
        \bm B
        \bm U^*(\hat{\bm\delta}-\bm{\delta}^\zeta)
        \leq
        K\sum_{k=1}^K
        (
        \hat{\bm c}^{(0)}-\hat{\bm c}^{(k)}-\bm\delta^{*(k)}
        )^\top
        \bm B
        (
        \hat{\bm c}^{(0)}-\hat{\bm c}^{(k)}-\bm\delta^{*(k)}
        )^\top,
    \end{align}
    \begin{align}
        \notag
        \{\widehat{\bm U}^*(\hat{\bm\delta}-\bm{\delta}^\zeta)\}^\top
        \bm B
        \widehat{\bm U}^*(\hat{\bm\delta}-\bm{\delta}^\zeta)
        \leq
        K\sum_{k=1}^K
        (
        \hat{\bm c}^{(0)}-\hat{\bm c}^{(k)}-\bm\delta^{*(k)}
        )^\top
        \bm B
        (
        \hat{\bm c}^{(0)}-\hat{\bm c}^{(k)}-\bm\delta^{*(k)}
        )^\top,
    \end{align}
    and
    \begin{align*}
        \notag
        \{{\bm U}^*&(\hat{\bm\delta}-\bm{\delta}^\zeta)\}^\top
        \bm B
        {\bm U}^*(\hat{\bm\delta}-\bm{\delta}^\zeta)
        \\
        &\leq
        K\sum_{k=1}^K
        (
        {\rm E}(\hat{\bm c}^{(0)}-\hat{\bm c}^{(k)}\mid\mathcal Z)-\bm\delta^{*(k)}
        )^\top
        \bm B
        (
        {\rm E}(\hat{\bm c}^{(0)}-\hat{\bm c}^{(k)}\mid\mathcal Z)-\bm\delta^{*(k)}
        )^\top.
    \end{align*}
    These three bounds yield that
    \begin{align*}
        &{\rm E}(\|\hat\beta_{PC}^{(0)}-\tilde\beta^{(0)}\|_{\widehat C^{(0)}}^2\mid\mathcal Z)
        \\
        &=
        {\rm E}
        (\|
        \bm\phi^\top(\bm U^*-\widehat{\bm U}^*)
        (\hat{\bm\delta}-\bm{\delta}^\zeta)
        -
        \bm\phi^\top\bm U^*
        \{{\rm E}(\hat{\bm\delta}\mid\mathcal Z)-\bm{\delta}^\zeta\}
        \|_{\widehat C^{(0)}}^2
        \mid\mathcal Z
        )
        \\
        &\leq
        3{\rm E}
        (\|
        \bm\phi^\top\bm U^*
        (\hat{\bm\delta}-\bm{\delta}^\zeta)
        \|_{\widehat C^{(0)}}^2
        \mid\mathcal Z
        )
        +
        3{\rm E}
        (\|
        \bm\phi^\top\widehat{\bm U}^*
        (\hat{\bm\delta}-\bm{\delta}^\zeta)
        \|_{\widehat C^{(0)}}^2
        \mid\mathcal Z
        )
        +
        \\
        &\qquad
        3{\rm E}
        (\|
        \bm\phi^\top{\bm U}^*
        \{{\rm E}(\hat{\bm\delta}\mid\mathcal Z)-\bm{\delta}^\zeta\}
        \|_{\widehat C^{(0)}}^2
        \mid\mathcal Z
        )
        \\
        &\leq
        3K\big\{
        {\rm E}
        (\|
        \bm\phi^\top
        (\hat{\bm c}^{(0)}-\hat{\bm c}^{(k)}-\bm\delta^{*(k)})
        \|_{\widehat C^{(0)}}^2
        \mid\mathcal Z
        )
        +
        {\rm E}
        (\|
        \bm\phi^\top
        (\hat{\bm c}^{(0)}-\hat{\bm c}^{(k)}-\bm\delta^{*(k)})
        \|_{\widehat C^{(0)}}^2
        \mid\mathcal Z
        )
        +
        \\
        &\qquad
        {\rm E}
        (
        \|
        \bm\phi^\top
        \{
        {\rm E}(\hat{\bm c}^{(0)}-\hat{\bm c}^{(k)}\mid\mathcal Z)
        -\bm\delta^{*(k)}
        \}
        \|_{\widehat C^{(0)}}^2
        \mid\mathcal Z
        )
        \big\}
        \quad\text{(Lemma \ref{lemma:aUBUa})}
        \\
        &=
        O_p(n^{-1}\lambda^{-1/4}J^\xi+n^{-2}\lambda^{-1}\zeta^2J^\xi).
    \end{align*}
    The rest of the proof follows in the same fashion as proofs of 
    Propositions \ref{prop:hat_beta_c_estimation_err} and \ref{prop:hat_beta_c_pred_err}
    and is therefore omitted.
\end{proof}

\end{document}